\documentclass[12pt,draftcls,onecolumn]{IEEEtran}
\ifCLASSOPTIONcompsoc
  \usepackage[nocompress]{cite}
\else
  \usepackage{cite}
\fi
\usepackage[labelfont=bf, textfont=bf]{caption}
\usepackage{amsmath}
\usepackage[most]{tcolorbox}
\usepackage{amsfonts}
\usepackage{amssymb}
\usepackage{amsthm}
\usepackage{mathtools}
\usepackage{bm}
\usepackage{dsfont}
\usepackage{tabularx}
\usepackage{setspace}
\usepackage{hyperref}
\usepackage{bm}
\usepackage{fancybox}
\usepackage{booktabs}
\usepackage{multirow}
\usepackage{threeparttable}
\usepackage{makecell}
\usepackage{booktabs}
\usepackage[utf8]{inputenc}
\usepackage[procnumbered,ruled,vlined,linesnumbered]{algorithm2e}

\usepackage{xcolor}
\hypersetup{
	colorlinks,
	linkcolor={red!50!black},
	citecolor={blue!50!black},
	urlcolor={blue!80!black}
}
\allowdisplaybreaks[4]
\DeclareMathOperator*{\argmin}{arg\,min}
\DeclareMathOperator*{\argmax}{arg\,max}
\DeclareMathOperator{\Tr}{Tr}
\theoremstyle{plain}
\newtheorem{theorem}{Theorem}[section]%

\newtheorem{lemma}[theorem]{Lemma}
\newtheorem{proposition}[theorem]{Proposition}
\newtheorem{corollary}[theorem]{Corollary}
\newtheorem{mydef}[theorem]{Definition}
\newtheorem{myprob}{Problem}

\newcommand{\1}{\mathbf{1}}

\DontPrintSemicolon
\SetKw{KwAnd}{and}
\SetFuncSty{textsc}
\SetKwInOut{Input}{Input\ \ \ \ }

\SetKwInOut{Output}{Output}

\def\trace#1{\mathrm{Tr} \left(#1 \right)}
\newcommand\LL{\bm{L}}
\newcommand{\kh}[1]{\left(#1\right)}
\def\aa{\bm{a}}
\def\bb{\bm{b}}
\renewcommand\AA{\bm{A}}
\def\defeq{\stackrel{\mathrm{def}}{=}}
\newcommand\ee{\boldsymbol{\mathit{e}}}
\newcommand{\exactGreedy}{\textsc{ExactGreedy}}
\def\setof#1{\left\{#1  \right\}}
\def\sizeof#1{\left|#1  \right|}
\def\mG{\mathcal{G}}
\def\PP{\bm{P}}
\def\mV{\mathcal{V}}
\def\mE{\mathcal{E}}
\def\mX{\mathcal{X}}
\def\mY{\mathcal{Y}}
\def\mZ{\mathcal{Z}}
\def\mmod{~\mathrm{mod}~}
\def\LL{\bm{L}}
\def\WW{\bm{W}}

\def\DD{\bm{D}}
\def\ppi{\bm{\pi}}
\def\PPi{\bm{\Pi}}

\def\oo{\1}
\def\ee{\bm{e}}
\def\II{\bm{I}}
\def\LM{\LL^{\dagger}}
\def\LK{(\LL^{\dagger})_{\backslash k}}

\providecommand{\abs}[1]{\left|#1\right|}

\newcommand{\true}{{{\tt true}}}
\newenvironment{fminipage}%
{\begin{Sbox}\begin{minipage}}%
		{\end{minipage}\end{Sbox}\fbox{\TheSbox}}

\newenvironment{algbox}[0]{\vskip 0.2in
	\noindent
	\begin{center}
		\begin{fminipage}{0.465\textwidth}
		}{
		\end{fminipage}
	\end{center}
	\vskip 0.2in
}

\begin{document}
\title{Resistance Distances in Directed Graphs: Definitions, Properties, and Applications}
\markboth{IEEE Transactions on Information Theory}%
{Zhu MZ, Zhu LW, Li H, Li W, Zhang ZZ \MakeLowercase{\textit{et al.}}: Resistance Distances in Directed Graphs: Definitions, Properties, and Applications}
\author{Mingzhe~Zhu, Liwang~Zhu, Huan~Li, Wei~Li, Zhongzhi~Zhang~\IEEEmembership{Member,~IEEE}
	\thanks{This work was supported by the National Natural Science Foundation of China (Nos. 61872093 and U20B2051),  Shanghai Municipal Science and Technology Major Project  (Nos.  2018SHZDZX01 and 2021SHZDZX03),  ZJ Lab, and Shanghai Center for Brain Science and Brain-Inspired Technology. \textit{(Corresponding authors:~Huan~Li;~Wei~Li;~Zhongzhi~Zhang.)}}
	\thanks{Mingzhe Zhu, Liwang Zhu, Huan Li and Zhongzhi Zhang are with the Shanghai Key Laboratory of Intelligent Information Processing, School of Computer Science, Fudan University, Shanghai 200433, China. Zhongzhi~Zhang is also with the Shanghai Engineering Research Institute of Blockchains,  Fudan University, Shanghai 200433, China; and Research Institute of Intelligent Complex Systems, Fudan University, Shanghai 200433, China. Wei~Li is with the Academy for Engineering and Technology, Fudan University, Shanghai, 200433, China.	(e-mail:mzzhu21@m.fudan.edu.cn; 19210240147@fudan.edu.cn;
		huanli@cis.upenn.edu; fd\_liwei@fudan.edu.cn; zhangzz@fudan.edu.cn).}
}
\IEEEtitleabstractindextext{%
\begin{abstract}
Resistance distance has been studied extensively in the past years, with the majority of previous studies devoted to undirected networks, in spite of the fact that various realistic networks are directed. Although several generalizations of resistance distance on directed graphs have been proposed, they either have no physical interpretation or are not a metric. In this paper, we first extend the definition of resistance distance to strongly connected directed graphs based on random walks and show that the two-node resistance distance on directed graphs is a metric. Then, we introduce the Laplacian matrix for directed graphs that subsumes the Laplacian matrix of undirected graphs as a particular case, and use its pseudoinverse to express the two-node resistance distance, and many other relevant quantities derived from resistance distances. Moreover, we define the resistance distance between a vertex and a vertex group on directed graphs and further define a problem of optimally selecting a group of fixed number of nodes, such that their resistance distance is minimized. Since this combinatorial optimization problem is NP-hard, we present a greedy algorithm with a proved approximation ratio, and conduct experiments on model and realistic networks to validate the performance of this approximation algorithm.

\end{abstract}

}
\maketitle
\IEEEdisplaynontitleabstractindextext
\IEEEpeerreviewmaketitle
\section{Introduction}\label{sec:introduction}

\IEEEPARstart{N}{etwork} science is a cornerstone in the study of realistic complex systems ranging from biologic to social systems. One of the most powerful tools for network science is electrical networks, which have led to great success in both algorithmic and practical aspects of complex networks~\cite{DoSn84}. Given an undirected graph, its underlying electrical network is the network obtained by replacing every edge $e$ with weight $w(e)$ in $\mG $ with a resistor having conductance $1/w(e)$. A fundamental quantity of electrical networks is effective resistance, also called resistance distance~\cite{KlRa93}. For any pair of nodes $i$ and $j$, its effective resistance is defined as the potential difference between them when a unit current is injected at $i$ and extracted from $j$. It has been proved that the effective resistance is a distance metric~\cite{KlRa93}, which plays a pivotal role in characterizing network structure~\cite{ElSpVaJaKo11} and various dynamics taking place on networks~\cite{Ne03}. 

Since its establishment, resistance distance has become an important basis of algorithmic graph theory, based on which researchers have obtained landmark results for fast algorithms solving multiple key problems~\cite{LiPaYiZh20}, such as computing maximum flows and minimum cuts~\cite{ChKeMaSpTe11,LeRaSr13,KeLeOrSi14}, sampling random trees~\cite{MaStTa15}, solving traveling salesman problems~\cite{AnGh15}, and sparsifying graphs~\cite{SpSr08}. In addition to its theoretical significance, resistance distance has proven ubiquitous in numerous practical settings, including graph clustering~\cite{SpSr08}, collaborative recommendation~\cite{FoPiReSa07}, graph embedding~\cite{CaZhCh18}, graph centrality~\cite{StZe89,BrFl05,ShYiZh18,LiPeShYiZh19}, link prediction~\cite{MaBeCu16}, and so on. 

Apart from the resistance distance itself, various graph invariants based on resistance distance have been defined and studied, such as the Kirchhoff index~\cite{KlRa93, GhBoSa08} and the multiplicative degree-Kirchhoff index~\cite{ChZh07}. The Kirchhoff index of a graph is the sum of effective resistances over all pairs of nodes, which has found broad applications in diverse fields~\cite{LiZh18, ZhXuYiZh22,YiYaZhZhPa22}. For example, it has been used to measure the overall connectedness of a network~\cite{TiLe10}, the global utility of social recommender networks~\cite{WoLiCh16}, as well as the robustness of the first-order consensus algorithm in noisy networks~\cite{PaBa14, QiZhYiLi19,YiZhPa20}. The multiplicative degree-Kirchhoff index of a graph is defined as a weighted sum of effective resistances of all node pairs. It is a multiple of the Kemeny's constant of the graph~\cite{ChZh07}, which spans a wide range of applications in various practical scenarios~\cite{Hu14,XuShZhKaZh20}. 

In view of the theoretical and practical importance of effective resistance and its related graph invariants, effective resistance has been studied extensively in the past decades~\cite{DoBu13, DoSibu18,ShZh19,ThYaNa19}. Particularly, a large volume of research has been devoted to resistance distance and their properties~\cite{ElSpVaJaKo11}. Most previous studies are intended for undirected networks, in spite of the fact that many of realistic networks are directed, including the World Wide Web, food webs, and social networks, among others. Although several existing studies touched on effective resistance for directed graphs~\cite{Ch19, YoScLe16a, YoScLe16b, SuSa22, BaBaGo20, BaBaGo22}, they have typically not tackled this question directly, since they have no physical interpretations and are not a distance metric. In this sense, there is a disconnect between the notion of effective resistance and directed graphs.

In this paper, we provide an in-depth study on effective resistances, their properties and applications in directed networks. First, based on the connection governing effective resistance and escape probability of random walks on undirected graphs~\cite{DoSn84}, we provide a natural generalization of effective resistance on undirected graphs to strongly connected directed graphs, which is shown to be a distance metric. We then define the Laplacian matrix for directed graphs and provide an examination on its properties, on the basis of which we provide expressions for two-node effective resistance, Kirchhoff index, and multiplicative degree-Kirchhoff index for directed graphs. Moreover, we introduce the notion of effective resistance between a node and a node group, and propose an NP-hard problem of selecting a set of fixed number of nodes, aiming at minimizing the sum of the effective resistance between the node group and all other nodes. We continue to prove that the objective function of the problem is monotone and supermodular, and develop a greedy algorithm to approximately solve this problem in cube time, which has a provable approximation guarantee. Finally, we carry on experiments on several model and realistic networks to evaluate this approximation algorithm.

\section{Preliminaries}

In this section, we briefly introduce some useful notations and tools for the convenience of the description of definitions, properties, and algorithms.

\subsection{Notations}

We use $\mathbb{R}$ to denote real number field, normal lowercase letters like $a,b,c$ to denote scalars in $\mathbb{R}$, calligraphic uppercase letters like $\mathcal{A},\mathcal{B}, \mathcal{C}$ to denote sets, bold lowercase letters like $\bm{a}, \bm{b}, \bm{c}$ to denote column vectors, and bold uppercase letters like $\bm{A}, \bm{B}, \bm{C}$ to denote matrices. We write $\bm{a}_{i}$ to denote the $i^{\text{th}}$ entry of vector $\bm{a}$ and $\bm{A}_{i,j}$ to denote  the $(i,j)^{\text{th}}$ entry of matrix $\bm{A}$ unless stated otherwise. We also write $\bm{A}_{i,:}$ to denote the $i^{\text{th}}$ row of $\bm{A}$ and $\bm{A}_{:,j}$ to denote the $j^{\text{th}}$ column of $\bm{A}$. 
For any matrix $\AA\in\mathbb{R}^{n\times n}$, we use $\AA^{\top}$ to denote its transpose satisfying $(\AA^\top)_{i,j} = \AA_{j,i}$, and we use $\trace{\AA}$ to denote the trace of the matrix $\AA$: $\trace{\AA}=\sum_{i=1}^n\AA_{i,i}$. For any vector $\aa\in\mathbb{R}^n$, we use $\mathrm{Diag}(\aa)$ to denote an $n$-by-$n$ diagonal matrix with its $i^{\text{th}}$ diagonal entry equalling $\aa_i$.

We write sets in matrix subscripts to denote submatrices. For example, $\bm{A}_{\mathcal{I}, \mathcal{J}}$ denotes the submatrix of $\bm{A}$  with row indices in $\mathcal{I}$ and column indices in $\mathcal{J}$. We write $\bm{A}_{\backslash i}$ to denote the submatrix of $\bm{A}$ obtained by removing the $i^{\text{th}}$ row and $i^{\text{th}}$ column of $\bm{A}$, and write $\bm{A}_{\backslash \mX}$ to denote the submatrix of $\bm{A}$ with rows and columns corresponding to indices in set $\mathcal{X}$ removed. For example, for an $n \times n$ matrix $\bm{A}$, $\bm{A}_{\backslash n}$ denotes the submatrix $\bm{A}_{1:n-1,1:n-1}$. It should be stressed that we use $\bm{A}^{-1}_{\backslash n}$ to denote the inverse of $\bm{A}_{\backslash n}$ instead of a submatrix of $\bm{A}^{-1}$.

For a matrix $\AA\in\mathbb{R}^{n\times n}$, we use $\lambda_1(\AA), \lambda_2(\AA),\ldots,\lambda_n(\AA)$ to denote its $n$ eigenvalues. Unless otherwise stated, the matrices considered in this paper are all real matrices. We introduce two types of generalized inverse for any matrix $\AA$.

\begin{mydef}\cite{Pe55}\label{def:moore}
	The Moore-Penrose inverse $\AA^{\dagger}$ of matrix $\AA$ is the matrix satisfying the following conditions:
	\begin{align*}
		&\bm{A}\bm{A}^{\dagger}\bm{A}=\bm{A},~~\bm{A}^{\dagger}\bm{A}\bm{A}^{\dagger}=\bm{A}^{\dagger},\\
		&\bm{A}\bm{A}^{\dagger}=\left(\bm{A}\bm{A}^{\dagger}\right)^\top,~~\bm{A}^{\dagger}\bm{A}=\left(\bm{A}^{\dagger}\bm{A}\right)^\top.
	\end{align*}
\end{mydef}

\begin{mydef}\cite{Er67}
	The group inverse $\AA^{\#}$ of $\AA$ is the matrix satisfying the following conditions:
	\begin{equation*}
		\bm{A}\bm{A}^{\#}\bm{A}=\bm{A},~~\bm{A}^{\#}\bm{A}\bm{A}^{\#}=\bm{A},~~\AA\AA^{\#}=\AA^{\#}\AA.
	\end{equation*}
\end{mydef}

In the sequel, unless otherwise noted, we refer to the Moore-Penrose inverse of a matrix simply as its pseudoinverse for conciseness. If a matrix $\AA$ commutes with its pseudoinverse, i.e., $\AA^\dagger\AA=\AA\AA^\dagger$, it is an EP-matrix~\cite{CaMe09}. In addition, we write $\bm{e}_i$, $\1$, $\bm{I}$ to denote, respectively, the $i^{\text{th}}$ standard basis vector, the all-ones vector and the identity matrix of appropriate dimension.

\subsection{Random Walks on Undirected Graphs}

Let $\mathcal{G} =(\mathcal{V},\mathcal{E},\bm{W})$ denote an undirected weighted graph on the node (vertex) set $\mathcal{V} = \{1, 2,\ldots,n\}$. Its weighted adjacent matrix $\bm{W} \in \mathbb{R}^{n\times n}$ is nonnegative and symmetric, the $(i,j)^{\text{th}}$ entry of which is defined as: $\bm{W}_{i,j} = \bm{W}_{j,i} > 0$ if $(i,j) \in \mathcal{E}$, and $\bm{W}_{i,j} = \bm{W}_{j,i} = 0$ otherwise. We use $\mathcal{N}(i)$ to denote the set of neighbors of node $i$. That is, for each node $j$ in $\mathcal{N}(i)$, there exists an edge $(i,j)\in\mE$. Then the degree of node $i\in\mathcal{V}$ is $d_i=
\sum_{j\in\mathcal{N}(i)}\WW_{i,j}=\sum_{j=1}^n\bm{W}_{i,j}$. The volume of $\mathcal{G}$, denoted by $d_{\mG}$, is defined as the sum of degrees of all nodes as $d_{\mG} = \sum_{i=1}^{n}d_i$. The Laplacian matrix $\LL$ of $\mG$ is defined as $\LL=\DD-\WW$, where $\DD={\rm Diag}(\WW\1)$ is the degree diagonal matrix of  $\mG$ with the $i^{\text{th}}$ diagonal element being the degree $d_i$ of node $i$.

The Laplacian matrix $\LL$ is symmetric and positive semidefinite. All its eigenvalues are non-negative, with a unique zero eigenvalue, and the null space of $\LL$ is $\{k\1|k\in\mathbb{R}\backslash\{0\}\}$. Since $\LL$ is not invertible,  its pseudoinverse $\LM$ is of great importance. As will be shown below, $\LM$ can be used to calculate various relevant quantities, such as the resistance distance and Kirchhoff index for electrical networks. Let $\bm{J}$ denote the matrix with all entries being ones, Then we have~\cite{GhBoSa08}
\begin{align}
    &\LL\LM=\LM\LL = \II-\frac{1}{n}\bm{J},\label{eq:LLLM}\\
    &\LM=\kh{\LL+\frac{1}{n}\bm{J}}^{-1}-\frac{1}{n}\bm{J}.\label{eq:LMexp}
\end{align} 
Note that for a general symmetric matrix,
it shares the same null space as its Moore-Penrose generalized
inverse~\cite{BeGr03}. Thus, the null space of $\LM$ is also $\{k\1|k\in\mathbb{R}\backslash\{0\}\}$.

The  normalized Laplacian matrix $\widetilde{\LL}$ of $\mG$ is defined as~\cite{Ch97}:
\begin{equation}\label{eq:normL}
\widetilde{\LL}\defeq\DD^{-1/2}\kh{\DD-\AA}\DD^{-1/2}=\DD^{-1/2} \LL \DD^{-1/2}.
\end{equation}
It is easy to verify that the normalized Laplacian matrix $\widetilde{\LL}$ is symmetric and positive semidefinite~\cite{HoJo12}, with all eigenvalues being nonnegative real numbers.

A random walk on graph $\mG$ is a Markov chain with transition probability matrix $\PP = \DD^{-1}\WW$. Here we assume that $\mG$ is a connected non-bipartite graph. Then the Markov chain is irreducible~\cite{AlFi02}, with a unique stationary distribution. Let $\ppi=(\ppi_{1},\ppi_{2},\ldots, \ppi_{n})^{\top}$ be the associated vector of stationary probabilities. Then $\ppi^{\top}\PP=\ppi^{\top}$ and $\1^\top\ppi=1$, obeying $\ppi_{i}=\frac{d_i}{\sum_{k}d_k}=\frac{d_i}{d_{\mG}}$ for $i=1,2,\ldots,n$. Moreover, this Markov chain on $\mG$ is reversible, satisfying 
$\ppi_{i}\PP_{i,j}=\ppi_{j}\PP_{j,i}$ for every pair of nodes  $i$ and $j$.

One of the most important quantities about random walks is the hitting time. The hitting time $H(i,j)$ from vertex $i$ to vertex $j$ is the expected time for a random walk starting from vertex $i$ visits vertex $j$ for the first time. The commute time $C(i,j)$ between vertex $i$ and vertex $j$ is the expected time taken by a random walk originating from vertex $i$ first reaches vertex $j$ and then returns to vertex $i$, namely, $C(i,j)=H(i,j)+H(j,i)$. The hitting time $H(i, \mathcal{X})$ from vertex $i$ to a vertex group $\mathcal{X}\subseteq\mathcal{V}$ is the expected time taken by a walk arrives at any node in $\mathcal{X}$ for the first time. The commute time $C(i,\mathcal{X})$ between vertex $i$ and vertex group $\mathcal{X}$ is the expected time needed by a random walk starting from vertex $i$ first visits some vertex in $\mathcal{X}$ and then returns to vertex $i$.

A key quantity based on hitting times for random walk on graph $\mG$ is the Kemeny’s constant $K(\mG)$, which is defined as the expected time required for a random walk starting from a vertex $i$ to a destination vertex chosen randomly according to
a stationary distribution $\ppi $ of random walks on $\mG$~\cite{Hu14}. In other words, $K(\mG)=\sum_{j=1}^{n}\ppi_jH(i,j)$, which is independent of the selection of starting vertex $i$~\cite{LeLo02}, obeying relation $\sum_{j=1}^{n}\ppi_jH(i,j)=\sum_{j=1}^{n}\ppi_jH(k,j)$ for an arbitrary pair of vertices $i$ and $k$. The quantity $K(\mG)$ is characterized by the eigenvalues of transition probability matrix $\PP$ and normalized Laplacian matrix $\widetilde{\LL}$~\cite{LeLo02}.
\begin{equation}\label{eq:Kemeny}
K(\mG)=\sum_{i=1\atop\lambda_i(\PP)\neq1}^n\frac{1}{1-\lambda_i(\PP)}= \sum_{i=1\atop\lambda_i(\widetilde{\LL})\neq 0}^n\frac{1}{\lambda_i(\widetilde{\LL})}.
\end{equation}

The Kemeny's constant has found applications in diverse areas~\cite{Hu14, XuShZhKaZh20}. First, it has been used to characterize the criticality~\cite{DeMeRoSaVa18} or connectivity~\cite{BeHe19} for a graph. It was also applied to measure the efficiency of user navigation through the World Wide Web~\cite{LeLo02}. Finally, it was exploited to quantify the performance of a class of noisy formation control protocols~\cite{JaOl19}, and the efficiency of robotic surveillance in network environments~\cite{PaAgBu15}. Very recently, nearly linear time algorithms for evaluating the Kemeny’s constant have been developed~\cite{ZhXuZh20,XuShZhKaZh20}.

The escape probability $P_{\mathrm{es}}(i,j)$ from vertex $i$ to vertex $j$ is the probability that a random walk starting at $i$ will reach $j$ before it returns to $i$. Analogously, the escape probability  $P_{\mathrm{es}}(i, \mathcal{X})$ from vertex $i$ to vertex set $\mathcal{X}$ is the probability that a random walk starting at $i$ will reach some node in set $\mathcal{X}$ before it returns to $i$.

\subsection{Harmonic Function, Electrical Network and Resistance Distance}

For graph $\mG=(\mV,\mE,\WW)$, a function $\phi: \mV\to\mathbb{R}$ defined on $\mG$ is called a harmonic function with boundary set $\mX\subset \mV$ if 
\begin{equation}\label{eq:harmonic}
	\sum_{j\in\mathcal{N}(i)}\PP_{i,j}\phi(j)=\phi(i).
\end{equation}
holds for every node $i\in\mV\backslash\mX$. The averaging in~\eqref{eq:harmonic} can be accounted for as an expectation after one jump of random walks. Thus, harmonic functions play an important role in the study of random walks and electrical networks, which has a close connection with random walks~\cite{DoSn84}. 

For any undirected weighted graph $\mG=(\mV,\mE,\WW)$, we can construct a corresponding electrical network by replacing each edge $(i,j)\in \mE$ with a resistor $r(i,j)=1/\WW_{i,j}$. Let $V(k)$ denote the electric potential at vertex $k$, and let $I(k)$ denote the amount of current injected into vertex $k$. If we apply a unit voltage between vertices $s$ and $t$, making $V(s) = 1$ and $V(t) = 0$, then the potential $V(k)$ at any vertex $k$ is a harmonic function with the boundary set $\{s,t\}$. Driven by the voltage, a current $I(s)$ will flow into the circuit from the outside source. The amount of current that flows depends upon the overall resistance in the circuit. Then, the resistance distance $\Omega(s,t)$ between vertices $s$ and $t$ is defined as $\Omega(s,t) = 1/I(s)$. The reciprocal $\zeta(s,t)$ of $\Omega(s,t)$ is called the effective conductance between $s$ and $t$. Note that if the voltage between $s$ and $t$ is multiplied by a constant, then the current $I(s)$ is multiplied by the same constant. Therefore, $\Omega(s,t)$ depends only on the ratio of the voltage between $s$ and $t$ to the current $I(s)$ flowing into the circuit.

The resistance distance $\Omega(i,j)$ is a remarkably important metric for measuring the similarity between vertices $s$ and $t$ on graph $\mG=(\mV,\mE,\WW)$~\cite{FoPiReSa07}. It can be expressed in terms of the entries of the pseudoinverse of Laplacian matrix $\LL$ as~\cite{KlRa93,FoPiReSa07}:
\begin{equation*}
	\Omega(i,j)=\left(\ee_i-\ee_j\right)^{\top}\LM\left(\ee_i-\ee_j\right).
\end{equation*}
The resistance distance $\Omega(i,j)$ can also be expressed in terms of the diagonal elements of the inverse for submatrices of $\LL$ as follows~\cite{IzKeWu13}:
\begin{equation}\label{eq:Omegaij}
\Omega(i,j)=\big(\LL^{-1}_{\backslash i}\big)_{j,j}=\big(\LL^{-1}_{\backslash j}\big)_{i,i}.
\end{equation}

Since electrical networks have been found to have interesting analogies of random walks in undirected graphs~\cite{DoSn84}, one can present a precise characterization of effective resistance in electrical networks in terms of random walks on corresponding graphs. For example, the resistance distance $\Omega(i,j)$ between a pair of vertices $i$ and $j$ in graph $\mG=(\mV,\mE,\WW)$ encodes their commute time $C(i,j)$~\cite{ChRaRuSmTi89}: 
\begin{equation*}
	d_{\mG}\Omega(i,j) = C(i,j).
\end{equation*}
Moreover, effective resistance can also be interpreted in terms of escape probability. 
It was shown in~\cite{DoSn84} that there exists an elegant connection between effective conductance $\zeta(i,j)$ and escape probability $P_{\mathrm{es}}(i,j)$ as $\zeta(i,j) = d_i\,P_{\mathrm{es}}(i,j)$. Considering $\zeta(i,j)=1/\Omega(i,j)$ and $d_i=d_{\mG}\ppi_i$, one can build the relation between resistance distance, escape probability and stationary distribution.
\begin{proposition}
For any pair of different vertices $i$ and $j$ in graph $\mG=(\mV,\mE,\WW)$,	
\begin{equation}\label{eq:omegaes}
		\Omega(i,j) = \frac{1}{d_{\mG}\ppi_iP_{\mathrm{es}}(i,j)}.
	\end{equation}
\end{proposition}

In addition to commute time and escape probability, many other quantities about random walks are related to some corresponding quantities of electrical networks. Let $\phi_{i,j}(k)$ denote the probability that a random walk starting from vertex $k$ will visit vertex $i$ before reaching $j$. Then, $\phi_{i,j}(\cdot)$ is a harmonic function with boundary set $\{i,j\}$. It has been known that $\phi_{i,j}(k)$ equals the voltage $V_k$ at vertex $k$ when a unit voltage is applied between $i$ and $j$~\cite{DoSn84}. By definition of escape probability, we have
\begin{equation*}
     P_{\mathrm{es}}(i,j) = 1-\sum\limits_{k\neq i} \bm{P}_{i,k} \phi_{i,j}(k).
\end{equation*}
Thus, for any pair of different  vertices $i$ and $j$,  the resistance distance $\Omega(i,j)$ can be alternatively expressed in the following way:
\begin{equation}\label{eq:omegaesPhi}
		\Omega(i,j) = \frac{1}{d_{\mG}\ppi_i \left(1-\sum\limits_{k\neq i} \bm{P}_{i,k} \phi_{i,j}(k)\right)}.
\end{equation}

We can also define the resistance distance $\Omega(i,\mX)$ between a vertex $i$ and a set of vertices $\mX$ in graph $\mG=(\mV,\mE,\WW)$. For this purpose, we treat $\mG$ as an electrical network, where all vertices in $\mX$ are grounded. Thus, those vertices in $\mX$ always have voltage $0$. The resistance distance $\Omega(i,\mX)$ is defined as the voltage of vertex $i$ when a unit current enters the network $\mG$ at vertex $i$ and leaves it at nodes $\mX$. For a random walk on an undirected graph $\mathcal{G}$, let $\phi_{i,\mX}(k)$ denote the probability that the walk starting at vertex $k$ reaches vertex $i$ before visiting any vertex in set $\mathcal{X}$. Then, $\phi_{i,\mX}(k)$ is equal to the voltage at vertex $k$ when a unit current is injected in vertex $i$ and extracted in nodes belonging to $\mathcal{X}$. It has been shown that this voltage equals $\ee_i^{\top}\LL_{\backslash \mX}^{-1}\ee_i$~\cite{LiPeShYiZh19}. Thus, we have
\begin{equation*}
\Omega(i,\mX)=\kh{\LL_{\backslash \mX}^{-1}}_{i,i},
\end{equation*} 
which is consistent with~\eqref{eq:Omegaij} when  $\mX$ includes only one node $j$.
The resistance distance $\Omega(i,\mX)$ can also be interpreted in terms of the escape probability $P_{\mathrm{es}}(i,\mX)$ of random walks as
\begin{equation*}
\Omega(i,\mX) = \frac{1}{d_{\mG}\ppi_iP_{\mathrm{es}}(i,\mX)}.
\end{equation*}

Besides the resistance distance itself, many other important quantities based on resistance distances have been defined and studied, such as the resistance distance of a single vertex or vertex set~\cite{StZe89,ShYiZh18}, the Kirchhoff index~\cite{KlRa93}, and the multiplicative degree-Kirchhoff index~\cite{ChZh07}.
\begin{mydef}\cite{BoFr13}\label{def:Omegai}
For a weighted undirected graph $\mG=(\mV,\mE,\WW)$,  the resistance distance of vertex $i$ is defined as 
	\begin{equation*}
		\Omega(i) = \sum_{j\in\mV}\Omega(i,j).
	\end{equation*}
\end{mydef}
\begin{mydef}\cite{LiPeShYiZh19}\label{def:OmegaX}
For a weighted undirected graph $\mG=(\mV,\mE,\WW)$, the resistance distance $\Omega(\mathcal{X})$ of a vertex group $\mathcal{X}\subseteq\mathcal{V}$ is
defined as
	\begin{equation*}
		\Omega(\mathcal{X}) = \sum_{j\in\mV} \Omega(i,\mathcal{X}).
	\end{equation*}
\end{mydef}

Both the resistance distance of a single vertex $i$ and the resistance distance of a vertex group $\mathcal{X}$ can be expressed in terms of the entries of the pseudoinverse of Laplacian matrix.
\begin{proposition}\label{thm:unOmegai}\cite{BoFr13}
For a weighted undirected graph $\mG=(\mV,\mE,\WW)$ with $n$ vertices, the resistance distance $\Omega(i)$ of any vertex $i\in\mV$ can be expressed as
	\begin{equation*}
		\Omega(i)=n\LM_{i,i} + 	\trace{\LM}.
	\end{equation*}
\end{proposition}
\begin{proposition}\label{thm:unOmegaX}\cite{LiPeShYiZh19}
For a weighted undirected graph $\mG=(\mV,\mE,\WW)$, the resistance distance $\Omega(\mathcal{X})$ of any vertex set $\mathcal{X} \subset \mV$ can be expressed as
\begin{equation*}
		\Omega(\mathcal{X}) = \Tr\left(\LL_{\backslash \mX}^{-1}\right).
	\end{equation*}
\end{proposition}
The resistance distance $\Omega(i)$ of vertex $i$ can be used to measure the importance of $i$, which is equal to the information centrality~\cite{StZe89,ShYiZh18}. Analogously, the resistance distance  $\Omega(\mX)$ has been applied to quantify the importance of nodes in $\mX$~\cite{LiPeShYiZh19}. 

We continue to introduce two quantities defined on the basis of resistance distances, the Kirchhoff index~\cite{KlRa93} and the multiplicative degree-Kirchhoff index~\cite{ChZh07}, both of which are graph invariants.  

\begin{mydef}\cite{KlRa93}
	For a weighted undirected graph $\mG=(\mV,\mE,\WW)$, the Kirchhoff index $R(\mathcal{G})$  of $\mG$ is defined as the sum of the resistance distances over all pairs of nodes in $\mV$: 
	\begin{equation*}
		R(\mathcal{G})=\sum_{i,j=1\atop i<j}^{N}\Omega(i,j).
	\end{equation*}
\end{mydef}


The multiplicative degree-Kirchhoff index~\cite{ChZh07} is a modification of the Kirchhoff index, which is defined as follows. 
\begin{mydef}\cite{ChZh07}
For a weighted undirected graph $\mG=(\mV,\mE,\WW)$, the multiplicative degree-Kirchhoff index $R^*(\mathcal{G})$  of $\mG$ is defined as the weighted sum of the resistance distances over all pairs of nodes in $\mV$: 
	\begin{equation*}
		R^*(\mathcal{G})=\sum_{i,j=1\atop i<j}^{N}(d_id_j)\Omega(i,j).
	\end{equation*}
\end{mydef}
It has been shown~\cite{ChZh07} that the multiplicative degree-Kirchhoff index $R^*(\mathcal{G})$ of a graph $\mathcal{G}$ is equal to $2d_{\mG}$ times the Kemeny constant of the graph. 


\section{Resistance distance on directed graphs}

In this section, we present a generalization of effective resistance for strongly connected directed graphs, which is a natural extension. Moreover, we introduce the Laplacian matrix for directed graphs and express the effective resistance in terms of the pseudoinverse of Laplacian matrix. Since the notion of electrical networks is inherently related to undirected graphs, we define the resistance distance for directed graphs based on random walks.

\subsection{Random Walks on Directed Graphs}

Let $\mG =(\mV, \mE, \WW)$ be a weighted directed graph (digraph) with vertex set $\mV = \{1, 2,\ldots,n\}$, edge set $\mE$, and nonnegative weighted adjacent matrix $\WW$. In general, $\WW$ is asymmetric, whose entry $\bm{W}_{i,j}$ is defined in the following way: $\bm{W}_{i,j} > 0$ if there is a directed edge (or arc) $\langle i, j\rangle$ in $\mE$ pointing to $j$ from $i$, $\bm{W}_{i,j}= 0$ otherwise. For any vertex $i\in\mV$, its out-degree is defined as $d_i^{+}=\sum_{j=1}^{n}\bm{W}_{i,j}$, and its in-degree is defined as $d_i^-=\sum_{j=1}^{n}\bm{W}_{j,i}$. Though $d^{+}_i$ is generally not equal to $d^-_i$, the relation $d_{\mG}=\sum_{i=1}^{n}d_i^{+}=\sum_{i=1}^{n}d_i^-=\sum_{i=1}^{n}\sum_{j=1}^{n}\bm{W}_{i,j}$ always holds. We call $d_{\mG}$ as the volume of the digraph $\mathcal{G}$. In the sequel, unless otherwise noted, we refer to the out-degree of a vertex simply as its degree and use $d_i$ for $d_i^+$ for conciseness.

For a digraph $\mG =(\mV, \mE, \WW)$, let $\bm{D}={\rm Diag}(\bm{W}\1)$ be its degree diagonal matrix, and define $\bm{P}=\bm{D}^{-1}\bm{W}$. By definition, $\bm{P}$ is the transition probability matrix of a Markov chain associated with a random walk on $\mathcal{G}$. At each time step, the random walk at its current state   $i$ jumps to a neighbor vertex $j$ with probability $\bm{P}_{i,j}=\bm{W}_{i,j}/d_i$. Throughout this paper we assume the graph $\mathcal{G}$ is strongly connected, i.e., every vertex in $\mV$ is reachable from every other vertex, implying that $\bm{P}$ is irreducible. Let $\ppi=(\ppi_{1},\ppi_{2},\ldots, \ppi_{n})^{\top}$ be the unique vector presenting the stationary distribution of the random walk on digraph $\mG$, satisfying $\ppi^{\top}\bm{P}=\ppi^{\top}$ and $\ppi^{\top}\1=\1$. Perron-Frobenius theory guarantees that $\ppi$ exists and its entries are strictly positive~\cite{HoJo94}. Let $\bm{\Pi}={\rm Diag}(\ppi)$ be the diagonal matrix with the entries of $\ppi$ on the diagonal.  

Note that for random walks on a digraph, the hitting time, the commute time, the escape probability, and the Kemeny's constant can also be defined as in the case of undirected graphs. Moreover,~\eqref{eq:Kemeny} still holds. In the case without incurring confusion, we represent relevant quantities for random walks on a digraph $\mG$ by using the same notations as those in undirected graphs. 

\subsection{Effective Resistance  between a Pair of Vertices}

Since electric networks cannot be constructed for digraphs, we define resistance distance on digraphs using the concept of escape probability for random walks by extending~\eqref{eq:omegaes} to digraphs.
\begin{mydef}\label{def:omega}
For a digraph $\mG=(\mV,\mE,\WW)$,  the resistance distance $\Omega(i,j)$ between any pair of vertices $i$ and $j$ is defined as
	\begin{equation*}
		\Omega(i,j) = \left\{
			\begin{aligned}
				&\frac{1}{d_{\mG}\ppi_iP_{\mathrm{es}}(i,j)},~~~~i\neq j,\\
				&0,~~\quad\quad\quad\quad\quad\quad i=j.
			\end{aligned}
	    \right.
    \end{equation*}
\end{mydef}

Recall that for random walks on a digraph $\mG$, $\phi_{i,j}(k)$ represents the probability that a walker starting from vertex $k$ will reach vertex $i$ before vertex $j$. We call the probability $\phi_{i,j}(k)$ as generalized voltage. It is easy to verify that the generalized voltage is still a harmonic function on directed graphs. Moreover, $P_{\mathrm{es}}(i,j)$ can be expressed in terms of  generalized voltages as $P_{\mathrm{es}}(i,j) = 1-\sum\limits_{k\neq i} \bm{P}_{i,k} \phi_{i,j}(k)$. Then as in the case of undirected graphs, the following relation holds 
\begin{equation}\label{eq:omegaesPhiD}
		\Omega(i,j) = \frac{1}{d_{\mG}\ppi_i \left(1-\sum\limits_{k\neq i} \bm{P}_{i,k} \phi_{i,j}(k)\right)}
\end{equation}
for $i \neq j$.

\begin{theorem}\label{thm:rcij}
For any pair of vertices $i$ and $j$ in digraph $\mG$,  the commute times $C(i,j)$ and the resistance $\Omega(i,j)$ obeys the following relation
	\begin{align}
		\Omega(i,j) = \frac{1}{d_{\mG}}C(i,j).
	\end{align}
\end{theorem}
\begin{proof}
For a random walk on digraph $\mG$, let $\tau$ be the first time for the walk starting at vertex $i$ returns to vertex $i$. And let $\sigma$ be the first time for a walk starting at vertex $i$ returns to $i$ after visiting vertex $j$. It is  known that the expectation of $\tau$ is $\mathrm{E}(\tau) = \frac{1}{\ppi_i}$~\cite{No97}. By definition, the expectation of $\sigma$ is $\mathrm{E}(\sigma) = C(i,j)$. It is obvious that $\tau \leq \sigma$ and the probability of $\tau = \sigma$ is exactly the escape probability $P_{\mathrm{es}}(i,j)$. Furthermore, for the case $\tau < \sigma$, after the first $\tau$ step jumpings, the walk will continue to jump from $i$ until it visits $j$ and then returns to $i$. Thus, we have $\mathrm{E}(\sigma-\tau) = \kh{1-P_{\mathrm{es}}(i,j)}\mathrm{E}(\sigma)$, which leads to
	\begin{equation*}
		P_{\mathrm{es}}(i,j)=\frac{1}{\ppi_iC(i,j)},
	\end{equation*}
Then, the assertion follows by Definition~\ref{def:omega}.
\end{proof}
Theorem~\ref{thm:rcij} extends the results for undirected graphs~\cite{ChRaRuSmTi89}.

On the basis of above-defined resistance distance between a pair of vertices on digraphs, we can further define the resistance distance for each vertex $i\in\mathcal{V}$ and related Kirchhoff indices for digraphs.
\begin{mydef}
	For a digraph $\mathcal{G}=\left(\mathcal{V},\mathcal{E},\WW\right)$, the resistance distance of vertex $i\in\mathcal{V}$ is 
	\begin{equation}
		\Omega(i)=\sum_{j\in\mathcal{V}}\Omega(i,j).
	\end{equation}
\end{mydef}
As for unweighted graphs, the resistance distance $\Omega(i)$ on directed graphs can be used to quantify the importance of vertex $i$. We will show later, the smaller the resistance distance $\Omega(i)$, the more important vertex $i$ is.

We proceed to extend the definitions of the Kirchhoff index and multiplicative degree-Kirchhoff index to digraphs. 
\begin{mydef}
	For a weighted digraph $\mG=(\mV,\mE,\WW)$ with stationary distribution $\ppi$ for random walks, the Kirchhoff index $\mathrm{R}(\mG)$ and the multiplicative Kirchhoff index $\mathrm{R}^*(\mG)$ are defined, respectively, as 
	\begin{equation}\label{eq:ki}
		\mathrm{R}(\mG) = \sum_{i,j=1\atop i<j}^n\Omega(i,j)
	\end{equation}
and
    \begin{equation}\label{eq:mki}
    	\mathrm{R}^*(\mG) = d_{\mG}^2\sum_{i,j=1\atop i<j}^n\ppi_i\ppi_j\Omega(i,j).
    \end{equation}
\end{mydef} 
It can be easily verified that the Kirchhoff index and multiplicative Kirchhoff index for digraphs subsume the two indices on undirected graphs as special cases.

\subsection{Effective Resistance  between a Vertex and a Vertex Group}

By using the escape probability, we can also define the resistance distance between a vertex and a group of vertices on digraphs.  
\begin{mydef}\label{thm:defrix}
For a digraph $\mathcal{G}=\left(\mathcal{V},\mathcal{E},w\right)$,  the resistance distance $\Omega(i,\mathcal{X})$ between a vertex $i\in\mathcal{V}$ and a nonempty vertex group $\mathcal{X}\subset\mathcal{V}$ with $i\notin\mathcal{X}$ is defined as
\begin{align}\label{eq:defrix}
\Omega(i,\mathcal{X}) = \frac{1}{d_{\mG}\ppi_iP_{\mathrm{es}}(i,\mX)}.
\end{align}
\end{mydef}
For random walks on a digraph $\mG$, we call the probability $\phi_{i, \mathcal{X}}(k)$ as the generalized voltage at vertex $k$. It is not difficult to verify that the escape probability $P_{\mathrm{es}}(i, \mathcal{X})$ satisfies $P_{\mathrm{es}}(i, \mathcal{X}) =  1-\sum\limits_{j\neq i} \bm{P}_{i,j} \phi_{i,\mathcal{X}}(j) $, we have 
\begin{equation}\label{eq:omegaesiXPhiD}
		\Omega(i, \mathcal{X}) = \frac{1}{d_{\mG}\ppi_i \left( 1-\sum\limits_{j\neq i} \bm{P}_{i,j} \phi_{i,\mathcal{X}}(j) \right)}.
\end{equation}

\begin{theorem}\label{thm:rciX}
For a vertex $i$ and a set $\mathcal{X}$ of vertices in digraph $\mG$,  the commute times $C(i,\mathcal{X})$ and the resistance distance $\Omega(i,\mathcal{X})$ obeys the following relation
	\begin{align}
		\Omega(i,\mathcal{X}) = \frac{1}{d_{\mG}}C(i,\mathcal{X}).
	\end{align}
\end{theorem}
\begin{proof}
The proof is similar to that of Theorem~\ref{thm:rcij}.
\end{proof}

We now introduce the notion of resistance distance of a vertex group $\mathcal{X}$.
\begin{mydef}\label{thm:centdef}
For a vertex group $\mathcal{X}\subseteq\mathcal{V}$ in a connected directed weighted graph $\mathcal{G} =(\mathcal{V}, \mathcal{E},\WW)$, its resistance distance $\Omega(\mathcal{X})$ is defined as
\begin{equation*}
\Omega(\mathcal{X}) = \sum_{i=1}^n \Omega(i,\mathcal{X}). 
\end{equation*}
\end{mydef}
By Definition~\ref{thm:centdef}, the resistance distance $\Omega(\mathcal{X})$ is the sum of the resistance distance $\Omega(i,\mX)$ between vertex group $\mX$ and all vertices $i\in \mV$. As will be shown later, $\Omega(\mathcal{X})$ can be used to identify the importance of the vertices in $\mathcal{X}$ as a group. The smaller the quantity $\Omega(\mathcal{X})$, the more important the vertices in $\mathcal{X}$. Thus, the reciprocal $1/\Omega(\mathcal{X})$ of  $\Omega(\mathcal{X})$ is a group centrality for digraphs.

\section{Laplacian Matrix and Effective Resistance for Digraphs}

In this section, we introduce the Laplacian matrix for digraphs and study its properties by using the tools of matrix analysis. Then, we represent resistance distances and their associated quantities in terms of the pseudoinverse of Laplacian matrix for digraphs. Moreover, we will prove that two-node resistance distance defined in the proceeding section is a distance metric.  

\subsection{Definition and Properties of Laplacian Matrix }

We here introduce the notion of the Laplacian matrix for digraphs, and present a detailed analysis for its properties and those for its pseudoinverse. 

\begin{mydef}\label{Def:dilap}
For a weighted digraph $\mG=(\mV,\mE,\WW)$, let $\PP$ denote its transition probability matrix and let $\PPi=\mathrm{Diag}(\ppi)$ denote the diagonal matrix with the stationary probabilities on the diagonal. Then, the Laplacian matrix $\LL$ of $\mG$ is defined as 
	\begin{equation}\label{eq:dilap}
		\LL\defeq d_{\mG}\PPi(\II-\PP).
	\end{equation}
\end{mydef}
It is easy to verify that $\LL=\bm{D}-\bm{W}$ when the graph $\mathcal{G}$ is undirected. Thus, the Laplacian matrix for digraphs is a natural extension of that for undirected graphs. Note that in~\cite{BoRaZh11,LiZh12, LiZh13}, the Laplacian matrix for digraphs has been defined as $\PPi(\II-\PP)$, which does not include the Laplacian matrix for undirected graphs as a particular case, and is thus different from that in Definition~\ref{Def:dilap}. As we will show below, using the Laplacian matrix in Definition~\ref{Def:dilap}, the effective resistances and their relevant quantities have the same expression forms as those for undirected graphs.  

For a digraph, we can also define the normalized Laplacian matrix $\widetilde{\LL}$ as 
\begin{equation}\label{eq:dinormL}
	\widetilde{\LL}\defeq \frac{\PPi^{-1/2}}{\sqrt{d_{\mG}}}\LL\frac{\PPi^{-1/2}}{\sqrt{d_{\mG}}}=\PPi^{1/2}(\II-\PP)\PPi^{-1/2}.
\end{equation}
It is easy to verify that~\eqref{eq:dinormL} is reduced to~\eqref{eq:normL} when $\mathcal{G}$ is undirected. 
The normalized Laplacian matrix $\widetilde{\LL}$ was  previously introduced in~\cite{BoRaZh11,LiZh12,LiZh13}.

We next show that the Laplacian matrix given by Definition~\ref{Def:dilap} has some typical properties or identical expressions as those for undirected graphs.
\begin{lemma}\label{thm:rcL}
	The sum of each column and each row of $\LL$ equals $0$. Namely, $\LL\1= \bm{0}$ and $\1^{\top}\bm{L} = \bm{0}$.
\end{lemma}
\begin{proof}
We first prove  $\LL\1= \bm{0}$. By definition of  $\LL$, one has
	\begin{equation*}
		\bm{L}\1=d_{\mG}\left(\ppi-\bm{\Pi}\PP\1\right)=d_{\mG}\left(\ppi-\bm{\Pi}\1\right)=\bm{0},
	\end{equation*}
where the fact that $\PP$ is a row stochastic matrix has been used.
Similarly, we have
	\begin{equation*}
\1^{\top}\LL=d_{\mG}\left(\ppi^{\top}-\1^{\top}\PPi\PP\right)=d_{\mG}\left(\ppi^{\top}-\ppi^{\top}\PP\right).
	\end{equation*}
	Considering $\ppi^{\top}\PP=\ppi^{\top}$ we derive $\1^{\top}\LL=\bm{0}$.
\end{proof}
\begin{lemma}\label{zerosp}
	The null space of $\LL$ is $\{k\1|k\in\mathbb{R}\backslash\{0\}\}$.
\end{lemma}
\begin{lemma}\label{thm:invertL-1}
The matrix $\LL - \frac{1}{n}\1\1^{\top}$ is invertible.
\end{lemma}
\begin{proof}
We prove the lemma by contradiction. Suppose that matrix $\LL - \frac{1}{n}\1\1^{\top}$ is not invertible, then there exists a nonzero vector $\bm{x}$ satisfying $\left(\bm{L} - \frac{1}{n}\1\1^{\top}\right)\bm{x} = \bm{0}$. Multiplying both sides by $\1^{\top}$  gives
\begin{equation*}
		\1^{\top}\left(\bm{L} - \frac{1}{n}\1\1^{\top}\right) \bm{x} = \bm{0}.
\end{equation*}
Considering $\1^{\top}\LL=\bm{0}$ and $\1^{\top}\1=n$ we obtain
\begin{equation}\label{eq:invertL-1}
		\1^{\top}\bm{x}=\bm{0}.
\end{equation}
On the other hand, by applying~\eqref{eq:invertL-1} we have $\left(\bm{L} - \frac{1}{n}\1\1^{\top}\right)\bm{x} = \bm{L}\bm{x}=\bm{0}$. According to Lemma~\ref{zerosp}, we have
	\begin{equation}\label{eq:invertL-2}
		\bm{x} = k\1, k\in \mathbb{R}\setminus{\{0\}}.
	\end{equation}
Both~\eqref{eq:invertL-1} and~\eqref{eq:invertL-2} lead to a contradiction. Therefore, $\left(\bm{L} - \frac{1}{n}\1\1^{\top}\right)$ is invertible.
\end{proof}
\begin{proposition}\label{thm:psuformula}
Let $\LM$ be the pseudoinverse of the Laplacian matrix $\LL$. Then, $\LM= \left(\bm{L} - \frac{1}{n}\1\1^{\top}\right)^{-1} + \frac{1}{n}\1\1^{\top}$.
\end{proposition}
\begin{proof}
Let $\bm{M}=\left(\bm{L} - \frac{1}{n}\1\1^{\top}\right)^{-1} + \frac{1}{n}\1\1^{\top}$. To prove $\LM=\bm{M}$, we will check the Penrose conditions. First, we prove that $\LL\bm{M}$ is symmetric. From Lemma~\ref{thm:rcL},
	\begin{align}\label{eq:lm}
		\bm{L}\bm{M} =&\bm{L}\left(\bm{L} - \frac{1}{n}\1\1^{\top}\right)^{-1} + \frac{1}{n} \bm{L}\1\1^{\top} \notag \\
		=& \left[\left(\bm{L} - \frac{1}{n}\1\1^{\top}\right) + \frac{1}{n}\1\1^{\top}\right]\left(\bm{L} - \frac{1}{n}\1\1^{\top}\right)^{-1}\notag\\
		=&\bm{I}+\frac{1}{n}\1\1^{\top}\left(\bm{L} - \frac{1}{n}\1\1^{\top}\right)^{-1}.
	\end{align} 
Again using Lemma~\ref{thm:rcL}, we have
	\begin{equation*}
		\1^{\top} \left(\bm{L}-\frac{1}{n}\1\1^{\top}\right) = -\1^{\top},
	\end{equation*}
which implies
	\begin{equation}\label{invtmp1}
		-\1^{\top} = \1^{\top} \left(\bm{L}-\frac{1}{n}\1\1^{\top}\right)^{-1}.
	\end{equation}
Both~\eqref{eq:lm} and~\eqref{invtmp1} lead to
\begin{align}\label{eq:lmn}
		\bm{L}\bm{M} =\II- \frac{1}{n} \1\1^{\top}.
	\end{align}
Thus, $\LL\bm{M}$ is symmetric. In a similar way, we next prove that $\bm{M}\LL$ is symmetric. From
	\begin{equation*}
		\left(\bm{L}-\frac{1}{n}\1\1^{\top}\right)\1= -\1,
	\end{equation*}
we have
	\begin{equation*}
		-\1= \left(\bm{L}-\frac{1}{n}\1\1^{\top}\right)^{-1} \1.
	\end{equation*}
	Therefore 
	\begin{align}\label{eq:ml}
		\bm{M}\bm{L} =&  \left[ \left(\bm{L} - \frac{1}{n}\1\1^{\top}\right)^{-1} + \frac{1}{n}\1\1^{\top} \right] \bm{L} \notag \\
		=& \left(\bm{L} - \frac{1}{n}\1\1^{\top}\right)^{-1} \left[\left(\bm{L} - \frac{1}{n}\1\1^{\top}\right) + \frac{1}{n}\1\1^{\top}\right]\notag\\
		=& \II- \frac{1}{n}\1\1^{\top},
	\end{align}
indicating that $\bm{M}\LL$ is symmetric. We continue to check  $\LL\bm{M}\LL=\LL$. Based on~\eqref{eq:lmn}, we have
	\begin{equation*}
		\bm{L}\bm{M}\bm{L} = \left(\II- \frac{1}{n}\1\1^{\top}\right) \bm{L} = \bm{L}.
	\end{equation*}
The final step is to prove  $\bm{M}\LL\bm{M}=\bm{M}$.~\eqref{eq:ml} means
	\begin{equation*}
		\bm{M}\bm{L}\bm{M}  = \left(\bm{I}-\frac{1}{n}\1\1^{\top} \right)\bm{M} = \bm{M} - \frac{1}{n}\1\1^{\top}\bm{M}.
	\end{equation*}
The term $\frac{1}{n}\1\1^{\top}\bm{M}$ is evaluated as
	\begin{align*}
		\frac{1}{n}\1\1^{\top}\bm{M}=&\frac{1}{n}\1\1^{\top} \left(\bm{L} - \frac{1}{n}\1\1^{\top}\right)^{-1}+ \frac{1}{n}\1\1^{\top} \times \frac{1}{n}\1\1^{\top}\\
		=&-\frac{1}{n}\1\1^{\top} + \frac{1}{n}\1\1^{\top}=\bm{0}.
	\end{align*}
Thus, we have
	\begin{equation*}
		\bm{M}\bm{L}\bm{M}=\bm{M}.
	\end{equation*}
Since $\bm{M}$  satisfies the Penrose conditions, it equals the pseudoinverse $\LM$ of $\LL$.
\end{proof}
\begin{corollary}\label{thm:rcL+}
The sum of each column and each row of $\LL^{\dagger}$ equals $0$. Namely, the following equalities hold
\begin{enumerate}
\item $\LL^{\dagger} \1= \bm{0}$;
\item $\1^{\top} \LL^{\dagger} = \bm{0}$;
\item $\sum\limits_{i,j\in\mathcal{V}} \LL^{\dagger}_{i,j} = 0$.
	\end{enumerate}
\end{corollary}
\begin{corollary}
The Laplacian matrix $\LL$ is EP-matrix. 
\end{corollary}
\begin{proof}
Combining~\eqref{eq:lmn} and~\eqref{eq:ml} gives $\LL\LM=\LM\LL$.
\end{proof}
Since $\LL$ is EP-matrix, its Moore-Penrose inverse $\LM$ equals its group inverse $\LL^{\#}$~\cite{CaMe09}.
\begin{lemma}\label{thm:invertL}
For any nonempty node set	 $\mX\subsetneqq\mathcal{V}$,  the submatrix $\LL_{\backslash \mX}$ is invertible and the inverse  $\LL_{\backslash \mX}^{-1}$ is entrywise nonnegative.
\end{lemma}
\begin{proof}
Clearly, the diagonal matrix $\bm{\Pi}$ is invertible, and its inverse $\bm{\Pi}^{-1}={\rm Diag}\left(\left(\frac{1}{\ppi_{1}},\frac{1}{\ppi_{2}},\ldots,\frac{1}{\ppi_{n}}\right)^{\top}\right)$. Define
	\begin{equation}\label{eqT}
		\bm{T}\defeq\frac{1}{d_{\mG}}\left[\lim\limits_{t\to\infty}\sum_{i=0}^{t}\left(\bm{P}_{\backslash \mX}\right)^{i}\right]\bm{\Pi}_{\backslash\mX}^{-1}.
	\end{equation}
	Since $\lVert\bm{P}_{\backslash \mX}\rVert<1$, we have
	\begin{equation*}
		\lim\limits_{t \to \infty}\left(\bm{P}_{\backslash \mX}\right)^t = 0,
	\end{equation*}
	which leads to the convergence of the limit in~\eqref{eqT}:
	\begin{equation}\label{eq:lim}
\lim\limits_{t\to\infty}\sum_{i=0}^{t}\left(\bm{P}_{\backslash \mX}\right)^{i}=\left(\bm{I}-\bm{P}_{\backslash \mX}\right)^{-1}.
	\end{equation}
	Thus, the following statement holds
	\begin{equation*}
		\LL_{\backslash \mX}\bm{T}=\bm{\Pi}_{\backslash \mX}\left(\bm{I}-\bm{P}_{\backslash \mX}\right)\left[\lim\limits_{t\to\infty}\sum_{i=0}^{t}\left(\bm{P}_{\backslash \mX}\right)^{i}\right]\bm{\Pi}_{\backslash \mX}^{-1}=\bm{I},
	\end{equation*}
	which implies $\bm{T}=\LL_{\backslash \mX}^{-1}$. Due to the nonnegativity of $\frac{1}{d_{\mG}}$, $\bm{\Pi}^{-1}$ and $\bm{P}_{\backslash \mX}$, $\bm{T}$ defined by~\eqref{eqT} is entrywise nonnegative.
\end{proof}

\begin{lemma}\label{thm:L-1L+}
Let $\hat{\LL}=\begin{pmatrix}
		\LL_{\backslash k} & \bm{p} \\
		\bm{q}^{\top} & \LL_{k,k}
	\end{pmatrix}$
be a block matrix obtained from the Laplacian matrix $\LL$ associated with a directed graph by exchanging respectively, the $k^{\mathrm{th}}$ and $n^{\mathrm{th}}$ rows, and $k^{\mathrm{th}}$ and $n^{\mathrm{th}}$ columns of $\LL$   where $\bm{p}=\begin{pmatrix}
		\LL_{1:k-1,k} \\
		\LL_{k+1:n,k} 
	\end{pmatrix}$ and $\bm{q}=\begin{pmatrix}
	\LL_{k,1:k-1}^{\top} \\
	\LL_{k,k+1:n}^{\top} 
\end{pmatrix}$. Let $\hat{\LL}^{\dagger} = \begin{pmatrix}
		(\LL^{\dagger})_{\backslash k} & \bm{r} \\
		\bm{s}^{\top} & (\LM)_{k,k}
	\end{pmatrix}$ be the matrix obtained from pseudoinverse $\LL^{\dagger}$ of $\LL$ by performing the same operation on $\LL^{\dagger}$. Then the inverse of the $(n-1)\times(n-1)$ matrix $\LL_{\backslash k}$ exists and is given by 
	\begin{equation*}
		\LL_{\backslash k}^{-1}=\left(\II+ \1\1^{\top}\right) \!(\LM)_{\backslash k}\!\left(\II+ \1\1^{\top}\right)
		=\begin{pmatrix}
			\II&\! -\1
		\end{pmatrix}
		\hat{\LL}^{\dagger}
		\begin{pmatrix}
			\II\\
			-\1^{\top}
		\end{pmatrix}.
	\end{equation*}
\end{lemma}
\begin{proof}
From Lemma~\ref{thm:rcL}, we obtain $\LL_{\backslash k}\1+\bm{p}=\bm{0}$,  $\1^{\top}\LL_{\backslash k}+\bm{q}^{\top}=\bm{0}$, and $\LL_{k,k}=-\1^{\top}\bm{p}=\1^{\top}\LL_{\backslash k}\1$; while from Corollary~\ref{thm:rcL+}, we have $(\LM)_{\backslash k}\1+\bm{r}=\bm{0}$,  $\1^{\top}(\LM)_{\backslash k}+\bm{s}^{\top}=\bm{0}$, and $(\LM)_{k,k}=-\1^{\top}\bm{r}=\1^{\top}(\LM)_{\backslash k}\1$. Then, we derive
	\begin{align}\label{eq:NN}
\kh{\II+\oo\oo^{\top}}\LK\kh{\II+\oo\oo^{\top}}=&\LK+\LK\oo\oo^{\top}+\oo\oo^{\top}\LK+\oo\oo^{\top}\LK\oo\oo^{\top}\notag\\
		=&\LK-\bm{r}\oo^{\top}-\oo\bm{s}^{\top}+\oo\LM_{k,k}\oo^{\top}\notag\\
		=&\begin{pmatrix}
			\II& -\1
		\end{pmatrix}
		\hat{\LL}^{\dagger}
		\begin{pmatrix}
			\II\\
			-\1^{\top}
		\end{pmatrix}.
	\end{align}
Multiplying on the left and right sides of the term in the last line of~\eqref{eq:NN} by $\LL_{\backslash k}$ yields:
	\begin{align*}
		\LL_{\backslash k}\begin{pmatrix}
			\II & -\oo
\end{pmatrix}\hat{\LL}^{\dagger}\begin{pmatrix}
			\II \\
			-\oo^{\top}
		\end{pmatrix}\LL_{\backslash k}
		=&\begin{pmatrix}
			\LL_{\backslash k} & \bm{p}
		\end{pmatrix}\hat{\LL}^{\dagger}\begin{pmatrix}
			\LL_{\backslash k}\\
			\bm{q}^{\top}
		\end{pmatrix}\\
		=&\begin{pmatrix}
			\II & \bm{0}
		\end{pmatrix}\hat{\LL}\hat{\LL}^{\dagger}\hat{\LL}\begin{pmatrix}
			\II \\
			\bm{0}
		\end{pmatrix}\\
		=&\begin{pmatrix}
			\II & \bm{0}
		\end{pmatrix}\hat{\LL}\begin{pmatrix}
			\II \\
			\bm{0}
		\end{pmatrix}
		=\LL_{\backslash k}.
	\end{align*}
Lemma~\ref{thm:invertL} implies that $\LL_{\backslash k}$ is invertible. Multiplying both sides of the above equation by matrix $\LL_{\backslash k}^{-1}$ on the left and right gives the desired result.
\end{proof}

\begin{corollary}\label{thm:corL-1L+}
Let $i$, $j$, and $k$ be three vertices in a weighted digraph $\mathcal{G} =(\mathcal{V}, \mathcal{E},\WW)$. If $i \neq k$ and $j \neq k$, then 
	\begin{equation*}
		\left(\LL^{-1}_{\backslash k}\right)_{i,j} = \LL^{\dagger}_{k,k} + \LL^{\dagger}_{i,j} - \LL^{\dagger}_{i,k} - \LL^{\dagger}_{k,j}.
	\end{equation*}
\end{corollary}
\begin{proof}
Applying Lemma~\ref{thm:L-1L+} we have
	\begin{align*}
		\left(\LL^{-1}_{\backslash k}\right)_{i,j} =& \bm{e}_i^{\top} \LL_{\backslash k}^{-1} \bm{e}_j
		= \bm{e}_i^{\top} \begin{pmatrix}
			\II& -\1
		\end{pmatrix}
		\hat{\LL}^{\dagger}
		\begin{pmatrix}
			\II\\
			-\1^\top
		\end{pmatrix} \bm{e}_j \\
		=& \bm{e}_i^{\top} (\LL^{\dagger})_{\backslash k} \bm{e}_j 
		- \bm{e}_i^{\top} \bm{r} \left(\1^{\top} \bm{e}_j\right) \\
		&- \left(\bm{e}_i^{\top} \1\right) \bm{s}^{\top} \bm{e}_j
		+ \LM_{k,k} \left( \bm{e}_i^{\top} \1\1^{\top} \bm{e}_j \right)\\
		=& \LL^{\dagger}_{i,j} - \LL^{\dagger}_{i,k} - \LL^{\dagger}_{k,j} + \LL^{\dagger}_{k,k},
	\end{align*}
	which completes the proof.
\end{proof}

\subsection{Expressions for Effective Resistance and Associated Quantities}

In this subsection, we provide expressions for two-node effective resistance,  effective resistance between a node and a node set, and their associated quantities, in terms of the entries of the pseudoinverse  $\LM$ of Laplacian matrix $\LL$ and its submatrices.

\begin{theorem}\label{thm:omegaL}
For a weighted digraph $\mG=(\mV,\mE,\WW)$, the resistance distance $\Omega(i,j)$ between a pair of vertices $i$ and $j$ is expressed as
\begin{align}\label{eq:rdef}
	\Omega(i,j) =(\bm{e}_i - \bm{e}_j)^{\top} \LL^{\dagger} (\bm{e}_i - \bm{e}_j)=\LL^{\dagger}_{i,i}+\LL^{\dagger}_{j,j}-\LL^{\dagger}_{i,j}-\LL^{\dagger}_{j,i}.
\end{align}
\end{theorem}
\begin{proof}
Note that when $i=j$, $\Omega(i,j)=0$ and $(\bm{e}_i - \bm{e}_j)^{\top} \LL^{\dagger} (\bm{e}_i - \bm{e}_j)=0$. Thus,~\eqref{eq:rdef} holds for $i=j$. We next show that~\eqref{eq:rdef} also holds for $i\neq j$.

According to the connection between escape probability and generalized voltages, we have
\begin{equation*}
    d_{\mG}\ppi_iP_{\mathrm{es}}(i,j)=d_{\mG}\ppi_i \left(1-\sum\limits_{k\neq i} \bm{P}_{i,k} \phi_{i,j}(k)\right),
\end{equation*}
which can be written in matrix-vector form as 
\begin{equation}\label{eq:mvOmega}
	d_{\mG}\ppi_iP_{\mathrm{es}}(i,j)=\left[d_{\mG}\PPi(\II-\PP)\bm{\phi}(i,j)\right]_i=\left[\LL\bm{\phi}(i,j)\right]_i.
\end{equation}
where $\bm{\phi}(i,j)=(\phi_{i,j}(1), \phi_{i,j}(2), \ldots, \phi_{i,j}(n))^{\top}$ with $\phi_{i,j}(i)$
$=1$ and $\phi_{i,j}(j)=0$. 
Recall that the generalized voltage $\phi_{i,j}(k)$ is a harmonic function, satisfying  
\begin{equation*}
\sum_{l\in\mathcal{N}(k)}\PP_{k,l}\phi_{i,j}(l)=\phi_{i,j}(k)
\end{equation*}
 for every $k\in\mV\backslash\{i,j\}$. From the above equation, we deduce
\begin{equation*}
\sum_{l\in\mathcal{N}(k)}d_{\mG}\ppi_k\PP_{k,l}\kh{\phi_{i,j}(k)-\phi_{i,j}(l)}=0,
\end{equation*}
which is recast in matrix-vector form as
\begin{equation}\label{eq:mvOmegaIJ}
\LL_{k,:}\bm{\phi}(i,j)=0,
\end{equation}
where $k \neq i$ and $k \neq j$. 
By~\eqref{eq:omegaesPhiD}, the resistance distance $\Omega(i,j)$ is represented as   
\begin{equation}\label{eq: Omega_ij}
\Omega(i,j)=\frac{1}{\left[\LL\bm{\phi}(i,j)\right]_i}.
\end{equation}
For brevity of the following proof, we define $\bm{\chi}(i,j)=\ee_i-\ee_j$. In order to obtain $\bm{\phi}(i,j)$, we introduce the following system of linear equations: 
 \begin{equation}\label{eq:LLx}
\LL\hat{\bm{\phi}}(i,j)=\alpha\bm{\chi}(i,j),
\end{equation}
where $\alpha>0$ is a constant.
Note that, any solution $\bm{\phi}(i,j)$ of the system of linear equations given by~\eqref{eq:mvOmegaIJ} is also a solution of system~\eqref{eq:LLx} with properly selected $\alpha$. Since $\LL$ is a singular matrix, system~\eqref{eq:LLx} has an infinite family of possible solutions. Moreover, if $\hat{\bm{\phi}}(i,j)$ is a solution of system~\eqref{eq:LLx}, then any solution of system~\eqref{eq:LLx} takes the form of $\hat{\bm{\phi}}(i,j)+\beta\1$, where $\beta$ is a real number. It is easy to verify that $\hat{\bm{\phi}}(i,j)=\alpha\LM\bm{\chi}(i,j)$ is a solution of system~\eqref{eq:LLx}. Then $\bm{\phi}(i,j)$ can be represented as $\bm{\phi}(i,j)=\hat{\bm{\phi}}(i,j)+\beta\1=\alpha\LM\bm{\chi}(i,j)+\beta\1$. From $\phi_{i,j}(i)=1$ and $\phi_{i,j}(j)=0$, we have
\begin{align}
    \alpha\ee_i^\top\LM\bm{\chi}(i,j)+\beta&=1,\label{eq:alphabeta1}\\
    \alpha\ee_j^\top\LM\bm{\chi}(i,j)+\beta&=0.\label{eq:alphabeta2}
\end{align}
Combining~\eqref{eq:alphabeta1} and~\eqref{eq:alphabeta2} gives 
\begin{align*}
    \alpha=\frac{1}{\bm{\chi}(i,j)^{\top}\LM\bm{\chi}(i,j)},~~~\beta=-\frac{\ee_j^\top\LM\bm{\chi}(i,j)}{\bm{\chi}(i,j)^{\top}\LM\bm{\chi}(i,j)}.
\end{align*}
Thus, we have
\begin{equation*}
\bm{\phi}(i,j) = \frac{\LM\bm{\chi}(i,j)}{\bm{\chi}(i,j)^{\top}\LM\bm{\chi}(i,j)}-\frac{\1\ee_j^\top\LM\bm{\chi}(i,j)}{\bm{\chi}(i,j)^{\top}\LM\bm{\chi}(i,j)}.
\end{equation*}
According to~\eqref{eq: Omega_ij}, $\Omega(i,j)$ can be expressed as
\begin{equation}\label{eq:omegaLp}
\Omega(i,j) =\frac{\bm{\chi}(i,j)^{\top}\LM\bm{\chi}(i,j)}{\left[\LL\LM\bm{\chi}(i,j)-\LL\1\ee^{\top}_j\LM\bm{\chi}(i,j)\right]_i}.
	\end{equation}
	From Lemma~\ref{thm:rcL}, we have
	\begin{equation}\label{eq:omegaLpTheq1}
		\Omega(i,j) = \frac{\bm{\chi}(i,j)^{\top}\LM\bm{\chi}(i,j)}{\left[\LL\LM\bm{\chi}(i,j)\right]_i}.
	\end{equation}
By using~\eqref{eq:lmn}, the numerator of right-hand side of~\eqref{eq:omegaLpTheq1} is 
    \begin{equation*}
    	\left[\LL\LM\bm{\chi}(i,j)\right]_i=\left[\left(\II-\frac{1}{n}\1\1^{\top}\right)\bm{\chi}(i,j)\right]_i=1.
    \end{equation*}
 Then, the desired result follows.
\end{proof}

The resistance distance $\Omega(i,j)$ between two vertices $i$ and $j$ can also be expressed in terms of the diagonal elements of the inverse for submatrices of $\LL$.
\begin{theorem}\label{thm:rl-1}
The resistance distance $\Omega(i,j)$ between any pair of vertices $i$ and $j$ in a weighted digraph $\mG=(\mV,\mE,\WW)$ can be represented as
	\begin{equation*}
		\Omega(i,j) = \left(\LL_{\backslash i}^{-1}\right)_{j,j} = \left(\LL_{\backslash j}^{-1}\right)_{i,i}.
	\end{equation*}
\end{theorem}
\begin{proof}
This is a direct conclusion from Theorem~\ref{thm:omegaL} and Corollary~\ref{thm:corL-1L+}.
\end{proof}

We proceed to prove that $\Omega(i,j)$  on a digraph $\mathcal{G}$  is a distance metric.
\begin{theorem}
The effective resistance in Definition~\ref{def:omega} is a metric.
In other words, for three vertices $i$, $j$, and $k$ in a digraph $\mG=(\mV,\mE)$, their resistance distances  satisfy the following three properties:
	\begin{enumerate}
		\item (Non-negativity) $\Omega(i,j)\geq0$ with equality  if and only $i=j$.
		\item (Symmetry) $\Omega(i,j) = \Omega(j,i)$.
		\item (Triangle inequality) $\Omega(i,j) \leq \Omega(i,k)+\Omega(k,j)$.
	\end{enumerate}
\end{theorem}
\begin{proof}
First, by definition $\Omega(i,j)=0$ when $i=j$. For $i \neq j$, it is easy to see that $\Omega(i,j)>0$ by using Lemma~\ref{thm:invertL} and Theorem~\ref{thm:rl-1}. 
	
From Theorem~\ref{thm:omegaL}, we have
	\begin{align*}
\Omega(i,j)=\LL^{\dagger}_{i,i}+\LL^{\dagger}_{j,j}-\LL^{\dagger}_{i,j}-\LL^{\dagger}_{j,i}
=\LL^{\dagger}_{j,j}+\LL^{\dagger}_{i,i}-\LL^{\dagger}_{j,i}-\LL^{\dagger}_{i,j}=\Omega(j,i),
	\end{align*}
which implies symmetry. 
	
We finally prove the triangle inequality. For  three arbitrary vertices $i$, $j$, and $k$ in  digraph $\mG$, we have
	\begin{align}\label{triangleeq}
		\Omega(i,k)+\Omega(k,j)-\Omega(i,j)
		=\left(\LL^{\dagger}_{k,k}+\LL^{\dagger}_{i,j}-\LL^{\dagger}_{k,j}-\LL^{\dagger}_{i,k}\right)
		+\left(\LL^{\dagger}_{k,k}+\LL^{\dagger}_{j,i}-\LL^{\dagger}_{k,i}-\LL^{\dagger}_{j,k}\right).
	\end{align}
By Corollary~\ref{thm:corL-1L+}, the left-hand side of~\eqref{triangleeq} is equal to $(\LL^{-1}_{\backslash k})_{i,j}+(\LL^{-1}_{\backslash k})_{j,i}$, which is nonnegative according to Lemma~\ref{thm:invertL}.
\end{proof}

\begin{proposition}\label{thm:OmegaiL}
The resistance distance $\Omega(i)$ for arbitrary vertex $i$ in a weighted digraph $\mathcal{G} =(\mathcal{V}, \mathcal{E},\WW)$ can be expressed as 
	\begin{equation*}
		\Omega(i) = n\LM_{i,i}+\trace{\LM}.
	\end{equation*}
\end{proposition}
\begin{proof}
	From Theorem~\ref{thm:omegaL} and Corollary~\ref{thm:rcL+}, we have
	\begin{align*}
		\Omega(i)  =& \sum_{j\in\mathcal{V}} \Omega(i,j)
		= \sum_{j\in\mathcal{V}}\kh{\LM_{i,i}+\LM_{j,j}-\LM_{i,j}-\LM_{j,i}}\\
		=& n\LM_{i,i}+\trace{\LM}+\sum_{j\in\mV}\LM_{i,j}+\sum_{j\in\mV}\LM_{j,i}\\
		=&n\LM_{i,i}+\trace{\LM},
	\end{align*}
 which finishes the proof.
  \end{proof}
When $\mathcal{G}$ is undirected,  Proposition~\ref{thm:OmegaiL} is reduced to Proposition~\ref{thm:unOmegai}.

The resistance distance $\Omega(i)$ can also be expressed by the trace or spectrum of submatrix $\LL_{\backslash i}$.
\begin{proposition}
The resistance distance $\Omega(i)$ for arbitrary vertex $i$ in a weighted digraph $\mathcal{G} =(\mathcal{V}, \mathcal{E},\WW)$ can be expressed as 
	\begin{equation*}
		\Omega(i)=\trace{\LL_{\backslash i}^{-1}}=\sum_{k=1}^{n-1} \frac{1}{\lambda_k(\LL_{\backslash i})}.
	\end{equation*}
\end{proposition}
\begin{proof}
	By Corollary~\ref{thm:rl-1}, we have
	\begin{align*}
		\Omega(i)  =& \sum\limits_{j\in\mathcal{V}} \Omega(i,j)
		= \sum\limits_{j\in\mathcal{V}\setminus{\{i\}}} \left(\LL_{\backslash i}^{-1}\right)_{j,j} 
		= \Tr\left(\LL_{\backslash i}^{-1}\right) 
		= \sum_{k=1}^{n-1} \frac{1}{\lambda_k(\LL_{\backslash i})}.
	\end{align*}
	This completes the proof.
\end{proof}

In the following, we show that these two Kirchhoff indices on directed graphs can be expressed by using the Laplacians. 

\begin{theorem}
For a weighted digraph $\mG=(\mV,\mE,\WW)$ with $n$ vertices and volume $d_{\mG}$, let $\LL$ and  $\widetilde{\LL}$ be, respectively, its Laplacian matrix and normalized Laplacian matrix. Then the Kirchhoff index $R(\mG)$  and the multiplicative Kirchhoff index $R^*(\mG)$ of digraph $\mG$ can be expressed as follows.
\begin{enumerate}
		\item\label{it:1} (Kirchhoff index)
		\begin{equation*}
			R(\mG) = n\trace{\LM}=n\sum_{i=1\atop \lambda_i(\LL)\neq 0}\frac{1}{\lambda_i(\LL)}.
		\end{equation*}
	    \item\label{it:2} (Multiplicative Kirchhoff index)  
	    \begin{equation*}
	    	R^*(\mG) = d_{\mG}K(\mG)=d_{\mG}\sum_{i=1\atop\lambda_i(\widetilde{\LL})\neq 0}^{n}\frac{1}{\lambda_i(\widetilde{\LL})}=d_{\mG}\trace{\widetilde{\LL}^{\dagger}}.
	    \end{equation*}
	\end{enumerate}
\end{theorem}
\begin{proof}
We first prove item~\ref{it:1}). From~\eqref{eq:rdef} and~\eqref{eq:ki}, we have
\begin{align}\label{eq:KG1}
		R(\mathcal{G}) =& \sum_{i,j=1\atop i<j}^n\Omega(i,j)=\frac{1}{2}\sum_{i,j\in\mathcal{V}}\Omega(i,j)
		= \frac{1}{2}\sum_{i,j\in\mV} \left(\LL^{\dagger}_{i,i} + \LL^{\dagger}_{j,j} - \LL^{\dagger}_{i,j} - \LL^{\dagger}_{j,i}\right).
	\end{align}
By Corollary~\ref{thm:rcL+}, the last sum term is evaluated by  
\begin{align}\label{eq:KG2}
\sum\limits_{i,j\in\mathcal{V}} \left(\LL^{\dagger}_{i,i} + \LL^{\dagger}_{j,j} - \LL^{\dagger}_{i,j} - \LL^{\dagger}_{j,i}\right)
		= \sum\limits_{i,j\in\mathcal{V}} \LL^{\dagger}_{i,i} + \sum\limits_{i,j\in\mathcal{V}} \LL^{\dagger}_{j,j}
		= 2n \Tr\left(\LL^{\dagger}\right).
	\end{align}
Considering~\eqref{eq:KG1},~\eqref{eq:KG2}, and the fact that $\trace{\LL^{\dagger}}=\sum_{\lambda_i(\LL)\neq 0}\frac{1}{\lambda_i(\LL)}$, we deduce 
	\begin{equation*}
		R(\mG) =n\trace{\LM}= n\sum_{\lambda_i(\LL)\neq 0}\frac{1}{\lambda_i(\LL)}.
	\end{equation*}
	
Next, we prove item~\ref{it:2}). Applying Corollary~\ref{thm:rcij} and considering the fact that $\sum_{j=1}^n\ppi_jH(i,j)$ is independent of  vertex $i$, we have
	\begin{equation*}
		\begin{aligned}
			R^*(\mG) =& d_{\mG}^2\sum_{i,j=1\atop i<j}^n\ppi_i\ppi_j\Omega(i,j)
			=d_{\mG}\sum_{i,j=1\atop i<j}^n\ppi_i\ppi_jC(i,j)\\
			=&d_{\mG}\sum_{i,j\in\mV}\ppi_i\ppi_jH(i,j)
			=d_{\mG}\sum_{j\in\mV}\ppi_jH(i,j)=d_{\mG}K(\mG).
		\end{aligned}
	\end{equation*}
 Since the Kemeny's constant can be expressed by the eigenvalues of $\PP$ and there is a one-to-one correspondence between the eigenvalues of $\PP$ and $\widetilde{\LL}$, we have
    \begin{equation*}
    	K(\mG) = \sum_{i=1\atop\lambda_i(\PP)\neq1}^n\frac{1}{1-\lambda_i(\PP)}= \sum_{i=1\atop\lambda_i(\widetilde{\LL})\neq0}^n\frac{1}{\lambda_i(\widetilde{\LL})}=\trace{\widetilde{\LL}^\dagger},
    \end{equation*}
which implies
\begin{equation*}
    	R^*(\mG)=d_{\mG}\sum_{i=1\atop\lambda_i(\widetilde{\LL})\neq0}^n\frac{1}{\lambda_i(\widetilde{\LL})}=d_{\mG}\trace{\widetilde{\LL}^\dagger}.
    \end{equation*}
This completes the proof.
\end{proof}

Finally, we use the diagonal elements of the inverse for submatrices of $\LL$ to represent the effective resistance $\Omega(i,\mX)$ between a vertex $i$ and a vertex group $\mX$, as well as the resistance distance $\Omega(\mX)$ of $\mX$. 

\begin{theorem}\label{GroupOmegaD}
In a digraph $\mG=(\mV,\mE,\WW)$, the effective resistance between a vertex $i\in \mV$ and a vertex group $\mX\subseteq\mV$ is
\begin{equation*}
\Omega(i,\mX) = \left(\LL_{\backslash \mX}^{-1}\right)_{i,i},
\end{equation*}
and  the resistance distance of  vertex set $\mathcal{X} \subset \mV$ is
\begin{equation*}
\Omega(\mathcal{X}) = \Tr\left(\LL_{\backslash \mX}^{-1}\right).
\end{equation*}
\end{theorem}
\begin{proof}
To prove the theorem, we first introduce some notations. Let  $\bm{\phi}(i,\mathcal{X})$ be the vector of generalized voltages $\phi_{i,\mathcal{X}}(k)$, $k=1,2,\ldots,n$, defined by
\begin{equation*}
		\bm{\phi}(i,\mathcal{X}) = \Big(\phi_{i,\mathcal{X}}(1), \phi_{i,\mathcal{X}}(2), \cdots, \phi_{i,\mathcal{X}}(n)\Big)^{\top}.
	\end{equation*}
Define $\bar{\bm{\phi}}(i,\mathcal{X})$ as
	\begin{equation*}
		\bar{\bm{\phi}}(i,\mathcal{X}) = \bm{\phi}(i,\mathcal{X})_{\backslash \{\mathcal{X} \cup \{i\}\}},
	\end{equation*}
namely, $\bar{\bm{\phi}}(i,\mathcal{X})$ is equal to $\bm{\phi}(i,\mathcal{X})$ with entries  corresponding to $\mathcal{X}\cup\{i\}$ being removed.
	
Without loss of generality, suppose $i$ is the last vertex of $\mathcal{V}\backslash\mathcal{X}$. Then, considering the fact that  $\phi_{i,\mathcal{X}}(i)=1$ and $\bm{P}_{i,i}=0$, we have 
\begin{equation}\label{eq:lxvt2}
		\begin{aligned}
\left[\LL_{\backslash{\mathcal{X}}}
			\begin{pmatrix}
				\bar{\bm{\phi}}(i,\mathcal{X}) \\
				1
			\end{pmatrix}\right]_{i}=&\left[\LL_{\backslash{\mathcal{X}}}
			\begin{pmatrix}
				\bar{\bm{\phi}}(i,\mathcal{X}) \\
				\phi_{i,\mathcal{X}}(i)
			\end{pmatrix}\right]_{i}
			\\
			=& d_{\mG}\ppi_{i}\left[-\sum_{
				j\in\mathcal{V}\backslash\mathcal{X}\atop
				j\neq i}\bm{P}_{i,j}\phi_{i,\mathcal{X}}(j)+\left(1-\bm{P}_{i,i}\right)\phi_{i,\mathcal{X}}(i)\right]\\
			=&d_{\mG}\ppi_{i} \left[ 1-\sum_{j\in\mathcal{V}\backslash\mathcal{X}\atop j\neq i} \bm{P}_{i,j} \phi_{i,\mathcal{X}}(j) \right] 
			\\=& d_{\mG}\ppi_{i} P_{\mathrm{es}}(i,\mathcal{X}).
		\end{aligned}
    \end{equation}
	
Furthermore, for any $k\neq i$, since $\bm{P}_{k,k}=0$ and $\sum_{j\in\mathcal{V}\backslash\mathcal{X}\atop
		j\neq k}\bm{P}_{k,j}\phi_{i,\mathcal{X}}(j)=\phi_{i,\mathcal{X}}(k)$, we have
	\begin{align*}
			\left[\LL_{\backslash{\mathcal{X}}}
			\begin{pmatrix}
				\bar{\bm{\phi}}(i,\mathcal{X}) \\
				1
\end{pmatrix}\right]_{k}=&\left[\LL_{\backslash{\mathcal{X}}}
			\begin{pmatrix}
				\bar{\bm{\phi}}(i,\mathcal{X}) \\
				\phi_{i,\mathcal{X}}(i)
			\end{pmatrix}\right]_{k}
			\\
			=& d_{\mG}\ppi_{k}\left[-\sum_{
				j\in\mathcal{V}\backslash\mathcal{X}\atop
				j\neq k}\bm{P}_{k,j}\phi_{i,\mathcal{X}}(j)+\left(1-\bm{P}_{k,k}\right)\phi_{i,\mathcal{X}}(k)\right]=0.
	\end{align*}
  Combining the above-obtained results, we get
\begin{align}\label{eq:lxvt}
		\LL_{\backslash \mX}
		\begin{pmatrix}
			\bar{\bm{\phi}}(i,\mathcal{X}) \\
			1
		\end{pmatrix}
		= \begin{pmatrix} \bm{0} \\ d_{\mG}\ppi_{i} P_{\mathrm{es}}(i,\mathcal{X}) \end{pmatrix},
	\end{align}
Multiplying $\LL_{\backslash \mX}^{-1}$ to the left on both sides of~\eqref{eq:lxvt} gives	\begin{align}\label{eq:lx-1t}
		\begin{pmatrix}
			\bar{\bm{\phi}}(i,\mathcal{X}) \\
			1
		\end{pmatrix}
		=
		\LL_{\backslash \mX}^{-1}
		\begin{pmatrix} \bm{0} \\ d_{\mG}\ppi_{i} P_{\mathrm{es}}(i,\mathcal{X}) \end{pmatrix}.
	\end{align}
	Considering the last row of~\eqref{eq:lx-1t}, we have
	\begin{align*}
		1 = \left(\LL^{-1}_{\backslash \mX}\right)_{i,i}d_{\mG}\ppi_{i} P_{\mathrm{es}}(i,\mathcal{X}),
	\end{align*}
	which yields
\begin{align}\label{eq:mrixc1}
\left(\LL^{-1}_{\backslash \mX}\right)_{i,i} = \frac{1}{d_{\mG} \ppi_{i}P_{\mathrm{es}}(i,\mathcal{X})}=\Omega(i,\mX),
	\end{align}
where the second equality is obtained by	Definition~\ref{thm:defrix}. 

The relation $\Omega(\mathcal{X}) = \Tr(\LL_{\backslash \mX}^{-1})$ follows directly from the definition of $\Omega(\mathcal{X})$ and the above expression for $\Omega(i,\mX)$.
\end{proof}
Note that Theorem~\ref{GroupOmegaD} generalizes the results for resistance distances of vertex group for undirected graph~\cite{LiPeShYiZh19}.

\section{Finding Vertex Group with Minimum Resistance Distance}

Identifying crucial vertex groups is a fundamental problem in data mining and graph applications~\cite{LaMe12,LuChRe16}. As an application, in this section we use the notion of resistance distance of a vertex group to find the most important group with fixed number of vertices. To this end, we first study the properties of the resistance distance $\Omega(\mX)$ as a function of set $\mX$ by showing that $\Omega(\mX)$ is monotone decreasing and supermodular. Then, we extend the NP-hard optimization problem for vertex group centrality in~\cite{LiPeShYiZh19} to digraphs: How to choose $1\leq k \leq n$ vertices forming set $\mX$, so that $\Omega(\mX)$ is minimized. Moreover, we propose a deterministic greedy approximation algorithm to solve the problem. Finally, we conduct experiments in model and real networks to evaluate our approximation algorithm.

\subsection{Properties of Resistance Distance for a Vertex Group }

Note that the resistance distance $\Omega(\mX)$ is a function of set $\mX$. In this subsection, we show that $\Omega(\mX)$ is monotone decreasing and supermodular. To this end, We first introduce the formal definitions for monotone and supermodular set functions. For simplicity, for any vertex $i$ and set $\mX$ of vertices,  we write $\mX+i$ to denote $\mX\cup\{i\}$ and $\mX-i$ to denote $\mX\backslash\{i\}$.

\begin{mydef}{(Monotonicity)}
A set function $f:2^{\mathcal{V}} \to \mathbb{R}$ is monotone decreasing if $f(\mX) \geq f(\mY)$ holds for all $\mX \subseteq\mY \subseteq \mV$.
\end{mydef}
\begin{mydef}{(Supermodularity)}
A set function $f : 2^{\mathcal{V}} \to \mathbb{R}$ is 	supermodular if $f(\mX) - f (\mX + i) \geq f (\mY) - f(\mY + i)$ holds for all
	$\mX \subseteq \mY \subseteq \mV$ and $i\in\mV$.
\end{mydef}
\begin{lemma}\label{thm:monov}
For an arbitrary pair of vertices $i$ and $j$ in graph $\mathcal{G}$, the generalized voltage $\phi_{j,\mathcal{X}}(i)$ is a monotone decreasing function of $\mathcal{X}$. Namely, for two nonempty vertex groups $\mathcal{X}\subseteq\mathcal{Y}\subseteq\mathcal{V}$,   
	\begin{align}
		\phi_{j,\mathcal{X}}(i) \geq \phi_{j,\mathcal{Y}}(i).
	\end{align}
\end{lemma}
\begin{proof}
Given an instantiation of a random walk, let $B_{i,j}(\mZ)$ denote the proposition that the random walk started at vertex $i$ and visited vertex $j$ without visiting any vertex in $\mZ$, and let $\mathrm{P}(B_{i,j}(\mZ))$  denote the probability that the proposition $B_{i,j}(\mZ)$ is true. Since for any vertex groups $\mX$ and $\mY$ such that $\mX\subset\mY$, $B_{i,j}(\mY)$ implies $B_{i,j}(\mX)$,  we have
	\begin{equation*}
		\phi_{j,\mathcal{X}}(i)=\mathrm{P}\left(B_{i,j}(\mX)\right)\geq\mathrm{P}\left(B_{i,j}(\mY)\right)=\phi_{j,\mathcal{Y}}(i),
	\end{equation*} 
leading to the result.
\end{proof}

In order to prove the monotonicity of function $\Omega(\mX)$, we introduce a notion, random detour ($i\to\mathcal{X}\to j$), which is a constrained random walk. Concretely, ($i\to\mathcal{X}\to j$) denotes a random walk starting from vertex $i$, that must visit some transit vertex in set $\mathcal{X}$, before it reaches vertex $j$ and stops. We write $H(i,\mathcal{X},j)$ to denote the expected number of jumping for a random detour ($i\to\mathcal{X}\to j$) to finish such a random walk. Note that when $\mX$ includes only one vertex, the defined random detour reduces to Definition 3 in~\cite{Ra13}.

\begin{lemma}\label{thm:monoh}
For an arbitrary pair of vertices $i$ and $j$ in graph $\mG=(\mV,\mE,\WW)$, $H(i,\mathcal{X},j)$ is a monotone decreasing function of the set $\mathcal{X}$ of transit vertices. In other words, for two vertex sets $\mathcal{X}\subseteq\mathcal{Y}\subseteq\mathcal{V}$,   
	\begin{align}
		H(i,\mathcal{X},j) \geq H(i,\mathcal{Y},j).
	\end{align}
\end{lemma}
\begin{proof}
  We first prove that for any vertex $k\in\mathcal{V}\backslash\{i,j\}$,
    \begin{equation}\label{eq:Hixj}
        H(i,\mathcal{X},j)\geq H(i,\mathcal{X}+k,j).
    \end{equation}
For this purpose, we use $A_s(t,\mZ)$ to denote the event that a random walk starting from vertex $s$ reaches vertex $t$ without visiting any vertex in set $\mZ$, and use $A_s(t,l)$ to denote the event that a random walk starting from vertex $s$ arrives at vertex $t$ without visiting vertex $l$. We write $\Psi_{s,t}(\mZ)$ to denote the number of jumping needed by a random walk starting from vertex $s$ to stop at vertex $t$ after passing through some vertex in $\mZ$. Then, $H(i,\mX,j)$ can be calculated as
    \begin{equation*}
    \begin{aligned}
        H(i,\mX,j)=\mathrm{E}(\Psi_{i,j}(\mX)) 
        =&\left(\mathrm{E}(\Psi_{i,j}(\{k\})|A_i(k,\mX+j))+H(j,\mX,j)\right)\phi_{k,\mX+j}(i)\phi_{j,\mX}(k)\\&+\mathrm{E}(\Psi_{i,j}(\mX)|A_i(j,k))\phi_{j,k}(i).
    \end{aligned}
    \end{equation*}
     Similarly, $H(i,\mathcal{X}+k,j)$ can be calculated as 
     \begin{equation*}
         \begin{aligned}
             H(i,\mX+k,j)=\mathrm{E}(\Psi_{i,j}(\mX+k))            =&\mathrm{E}(\Psi_{i,j}(\{k\})|A_i(k,\mX+j))\phi_{k,\mX+j}(i)\phi_{j,\mX}(k)\\&+\mathrm{E}(\Psi_{i,j}(\mX)|A_i(j,k))\phi_{j,k}(i).
         \end{aligned}
     \end{equation*}
Combining the above two relations, one has
    \begin{equation*}
    	\begin{aligned}    		
     H(i,\mathcal{X},j)-H(i,\mathcal{X}+k,j)=\phi_{k,\mX+j}(i)\phi_{j,\mX}(k)H(j,\mathcal{X},j).
    	\end{aligned}
    \end{equation*}
    Since $\phi_{k,\mX+j}(i)\phi_{j,\mX}(k)$ is nonnegative, we have $H(i,\mathcal{X},j)$
    $\geq H(i,\mathcal{X}+k,j)$. Suppose that $\mathcal{Y}\backslash\mathcal{X}=\{i_1,i_2,\ldots,i_q\}$, where $q=\lvert\mathcal{Y}\rvert-\lvert\mathcal{X}\rvert$. Then we can construct a series of vertex sets $\mathcal{Y}_0,\mathcal{Y}_1,\mathcal{Y}_2,\ldots\mathcal{Y}_q$ satisfying $\mathcal{Y}_0=\mathcal{X}$, $\mathcal{Y}_q=\mathcal{Y}$, and $\mathcal{Y}_{p}=\mathcal{Y}_{p-1}+i_p$ for $p=1,2,\ldots,q$. By using~\eqref{eq:Hixj}, we have
    \begin{align*}
H(i,\mathcal{X},j)=H(i,\mathcal{Y}_0,j)&\geq H(i,\mathcal{Y}_1,j)
\geq\ldots\geq H(i,\mathcal{Y}_q,j)=H(i,\mathcal{Y},j),
    \end{align*}
which completes the proof.
\end{proof}

\begin{proposition}\label{thm:monoc}
The resistance distance $\Omega(\mathcal{X})$ is a monotone decreasing and supermodular function of the set $\mathcal{X}$. In other words, for two nonempty vertex sets $\mathcal{X}\subseteq\mathcal{Y}\subseteq\mathcal{V}$ and any vertex $i\in\mathcal{V}$ with $i\notin\mathcal{Y}$,   
	\begin{align}\label{eq:monorc}
		\Omega(\mathcal{X}) \geq \Omega(\mathcal{Y})
	\end{align}
	and
	\begin{equation}\label{eq:supc}
		\Omega(\mathcal{X}) - \Omega\left(\mathcal{X}+i\right) \geq \Omega(\mathcal{Y}) - \Omega\left(\mathcal{Y}+i\right).
	\end{equation} 
\end{proposition}
\begin{proof}
According to Theorem~\ref{thm:rciX} and Definition~\ref{thm:centdef}, we have 
	\begin{align}\label{eq:OmegaComX}
		\Omega(\mathcal{X}) = \frac{1}{d_{\mG}}\sum\limits_{i=1}^n C(i,\mathcal{X}).
	\end{align}
Thus, we can alternatively prove this theorem by showing the monotonicity and supermodularity of the commute time $C(i,\mX)$.
Considering a special case $i=j$ of Lemma~\ref{thm:monoh}, we have
    \begin{align}\label{eq:monoc}
    	C(i,\mathcal{X}) \geq C(i,\mathcal{Y}).
    \end{align}
    Combining~\eqref{eq:OmegaComX} and~\eqref{eq:monoc} gives
    \begin{equation*}
    	\Omega(\mathcal{X})=\frac{1}{d_{\mG}}\sum_{i=1}^{n}C(i,\mathcal{X})
    	\geq\frac{1}{d_{\mG}}\sum_{i=1}^{n}C(i,\mathcal{Y})=\Omega(\mathcal{Y}),
    \end{equation*}
 which implies the monotonicity of $\Omega(\mX)$.

We continue to prove the supermodularity of $\Omega(\mX)$. As the notation $A_s(t,\mZ)$ defined in the proof of Lemma~\ref{thm:monoh}, for any pair of vertices $i$ and $j$ in $\mV$ and a set of vertices $\mX\subseteq\mV$, we introduce two notations $A_i(j,\mX)$ and $A_i(\mX,j)$, where $A_i(j,\mX)$ has been used in the proof of Lemma~\ref{thm:monoh} and $A_i(\mX,j)$ represents the event that a random walk starting from vertex $i$ reaches some vertex in $\mX$ without visiting vertex $j$. Then,  $\mathrm{P}\kh{A_i(j,\mX)}=\phi_{j,\mX}(i)$ and $\mathrm{P}\kh{A_i(\mX,j)}=1-\mathrm{P}\kh{A_i(j,\mX)}$ hold.
    
Let $\Psi_i(\mX)$ be the number of jumping required by a random walk originating from vertex $i$ visits some vertex in $\mX$ and comes back to $i$.  It is easy to see that the expectation $\mathrm{E}(\Psi_i(\mX))$ of the random variable $\Psi_i(\mX)$ is equal to $C(i,\mX)$. By definition, $C(i,\mX)$ can be calculated as
    \begin{equation*}
    	C(i,\mX)  =  \mathrm{E}\kh{\Psi_i(\mX)} 
    	=  \phi_{j,\mX}(i)\mathrm{E}\kh{\Psi_i(\mX)|A_i(j,\mX)}  +(1-\phi_{j,\mX}(i))\mathrm{E}\kh{\Psi_i(\mX)|A_i(\mX,j)},
    \end{equation*}
where $j\notin\mX$.
Similarly, $C(i,\mX+j)$ can be evaluated as
    \begin{equation*}
    		C(i,\mX+j)= \phi_{j,\mX}(i)\mathrm{E}\kh{\Psi_i(\mX+j)|A_i(j,\mX)}  
    		+(1-\phi_{j,\mX}(i))\mathrm{E}\kh{\Psi_i(\mX+j)|A_i(\mX,j)}.
    \end{equation*}
 Then, we have
    \begin{equation}\label{eq:CixmCixj}
    	\begin{aligned}
    		C(i,\mX)-C(i,\mX+j)
    		=&\phi_{j,\mX}(i)\mathrm{E}\kh{\Psi_i(\mX)-\Psi_i(\mX+j)|A_i(j,\mX)}\\
    		&+(1-\phi_{j,\mX}(i))\mathrm{E}\kh{\Psi_i(\mX)-\Psi_i(\mX+j)|A_i(\mX,j)}.
    	\end{aligned}
    \end{equation}
Since $\Psi_i(\mX)=\Psi_i(\mX+j)$ holds with probability $\mathrm{P}\kh{A_i(j,\mX)}$
$=\phi_{j,\mX}(i)$,~\eqref{eq:CixmCixj} can be rewritten as
    \begin{align*}
    	C(i,\mX)-C(i,\mX+j)
    	=&\phi_{j,\mX}(i)\mathrm{E}\kh{\Psi_i(\mX)-\Psi_i(\mX+j)|A_i(j,\mX)}\\
    	=&\phi_{j,\mX}(i)\left(H(j,\mX,i)-H(j,i)\right).
    \end{align*}
In a similar way, we obtain
    \begin{equation*}
    	C(i,\mY)-C(i,\mY+j)=\phi_{j,\mY}(i)\left(H(j,\mY,i)-H(j,i)\right).
    \end{equation*}
By Lemma~\ref{thm:monov} one has $\phi_{j,\mX}(i)\geq \phi_{j,\mY}(i)$, and by Lemma~\ref{thm:monoh} one has $H(j,\mX,i)\geq H(j,\mY,i)$. Then, we derive 
    \begin{equation}\label{eq:superCom}
    	C(i,\mathcal{X}) - C(i,\mathcal{X}+j) \geq C(i,\mathcal{Y}) - C(i,\mathcal{Y}+j),
    \end{equation}
Combining~\eqref{eq:OmegaComX} and~\eqref{eq:superCom} gives 
	\begin{align*}
		\Omega(\mathcal{X}) - \Omega\left(\mathcal{X}+i\right)
		=&\frac{1}{d_{\mG}}\sum_{i=1}^{n}C(i,\mathcal{X})-\frac{1}{d_{\mG}}\sum_{i=1}^{n}C(i,\mathcal{X}+i)\\
		=&\frac{1}{d_{\mG}}\sum_{i=1}^{n}\left[C(i,\mathcal{X})-C(i,\mathcal{X}+i)\right]\\
		\geq&\frac{1}{d_{\mG}}\sum_{i=1}^{n}\left[C(i,\mathcal{Y})-C(i,\mathcal{Y}+i)\right]\\
		=&\Omega(\mathcal{Y}) - \Omega\left(\mathcal{Y}+i\right),
	\end{align*}
which finishes the proof.
\end{proof}

\subsection{Problem Formulation and Algorithms}

Proposition~\ref{thm:monoc} shows that $\Omega(\mX)$ is a decreasing set function. In other words, the addition of any vertex into set $\mX$ will lead to a decrease of the resistance distance $\Omega(\mX)$. Then, we naturally raise the following problem: How to optimally select a set $\mX$ with $1 \leq k \ll n$ vertices in $\mV$, so that the resistance distance $\Omega(\mX)$ of the vertex set $\mX$ is minimized. Mathematically, the resistance distance minimization problem is formally stated as follows.

\begin{tcolorbox}
\begin{myprob}[Resistance Distance Minimization, RDM]\label{prob:rdm}
Given a weighted digraph $\mathcal{G} = (\mathcal{V}, \mathcal{E},\WW)$ with $n$ vertices and $m$ edges, and an integer $1 \leq k \ll n$, find a set $\mathcal{X}^*\subset \mathcal{V}$ of $k$ vertices such that the resistance distance $\Omega(\mathcal{X}^{*})$ is minimized. This set optimization problem is formulated as: 
            $\mathcal{X}^{*}\in \argmin\limits_{\mathcal{X}\subset\mathcal{V},\lvert\mathcal{X}\rvert=k}\Omega(\mathcal{X})$.
	\end{myprob}
\end{tcolorbox}

Problem~\ref{prob:rdm} is inherently a combinatorial problem. It can be solved accurately by the following na\"{\i}ve brute-force method. For each set $\mX$ of the $\binom{n}{k}$ possible subsets of vertices, compute the resistance distance of $\mX$. Then, output the set $\mX^*$ of vertices, which has the minim resistance distance. This method appears to be simple, but it is computationally impractical even for small-scale graphs, since its computation complexity is exponential, scaling with $n$ as $O\left(\binom{n}{k}\cdot n^3\right)$. In fact, Problem~\ref{prob:rdm} is NP-hard, since it is even NP-hard for 3-regular undirected graphs~\cite{LiPeShYiZh19}, for which there is a reduction from Problem~\ref{prob:rdm} to the problem of vertex cover.

To tackle the exponential complexity of a combinatorial optimization problem, one often resorts to greedy heuristic approaches. Due to the monotonicity and supermodularity of the objective function $\Omega(\cdot)$ for Problem~\ref{prob:rdm}, a simple greedy method is guaranteed to have a $(1 - \frac{k}{k-1}\cdot\frac{1}{e})$-approximation solution to this Problem~\cite{NeWoFi78}. This greedy method is as follows. Initially, the vertex set $\mX$ is set to be empty. Then $k$ vertices are added to $\mX$, each of which is chosen iteratively from set $\mV\backslash\mX$. In every iteration step $i$, vertex $v_i$ in the set $\mV\backslash\mX$ of candidate vertices is selected, such that the resistance distance $\Omega(\mX +v_i)$ is minimum among all $\Omega(\mX +v)$ with $v \in \mV\backslash\mX$. The iteration process terminates when $k$ vertices are chosen to be added to $\mX$. In each iteration, we need to compute the resistance distance $\Omega(\mX +v)$ for every candidate vertex $v$  in  $\mV\backslash\mX$, which involves matrix inversion. Suppose that directly inverting a matrix requires $O(n^3)$ time, the total computation complexity of the simple greedy method is $O(kn^4)$ for small $k$.

The $O(kn^4)$ complexity of the above simple greedy method is still very high. It can be reduced to $O(n^3)$ without sacrificing its accuracy. For this purpose, we first provide an expression of the marginal gain of each vertex $v\notin \mX$. 
\begin{lemma}\label{thm:margin}
For a vertex set $\mZ \subset \mV$ and a vertex $v\notin\mZ$ in a
 weighted digraph $G = (\mV, \mE, \WW)$, define $\Delta(\mZ,v)\defeq\Omega(\mZ) - \Omega(\mZ+v)$. Then,
\begin{equation*}
     \Delta(\mZ,v) = \frac{\ee_v^\top \LL_{\backslash\mZ}^{-2} \ee_v}{
		\ee_v^\top \LL_{\backslash\mZ}^{-1} \ee_v}.
\end{equation*}
\end{lemma}
\begin{proof}
 For a vertex $v\notin\mZ$, 
We write $\trace{\LL_{\backslash\mZ}}$
in block form as
\begin{equation*}
    \trace{\LL_{\backslash\mZ}} = \trace{
\begin{pmatrix}
	d_v & \aa^\top \\
	\bb & \LL_{\backslash(\mZ+v)}
\end{pmatrix}},
\end{equation*}
where $\AA = \LL_{\backslash(\mZ+v)}$, $\bm{a}^\top=\LL_{v,\mV\backslash(\mZ+v)}$, and $\bm{b}=\LL_{\mV\backslash(\mZ+v),v}$. 
By blockwise matrix inversion, we obtain
\begin{equation}\label{eq:block}
	\trace{\LL_{\backslash\mZ}^{-1}} \!=\!
	\trace{\!\begin{pmatrix}
		\frac{1}{s} \!\!&\!\!
		-\frac{1}{s} \aa^\top \AA^{-1} \\
		-\frac{1}{s} \AA^{-1} \bb \!\!&\!\!
		\AA^{-1} \!\!+\!\!
		\frac{1}{s}
		\AA^{-1}
		\bb \aa^\top
		\AA^{-1}
	\end{pmatrix}\!},
\end{equation}
where $s = d_u - \aa^\top \AA^{-1} \bb$.  Then   $\Delta(\mZ,v)$ can be expressed as
\begin{equation*}
	\begin{aligned}
        \Delta(\mZ,v) = &\Omega(\mZ)-\Omega(\mZ+v)\\
		=&\trace{\LL_{\backslash\mZ}^{-1}} - \trace{\LL_{\backslash(\mZ+v)}^{-1}}
		\\
		= & \kh{\LL_{\backslash\mZ}^{-1}}_{v,v}
		+ \sum\limits_{u\in \mV\setminus (\mZ + v)}
		\kh{ \kh{\LL_{\backslash\mZ}^{-1}}_{u,u} -
			\kh{\LL_{-(\mZ + v)}^{-1}}_{u,u}} \\
		= & \frac{1}{s} + \frac{1}{s}
		\trace{\AA^{-1} \bb \aa^\top \AA^{-1}}
		= \frac{1}{s} + \frac{1}{s} \aa^\top \AA^{-1} \AA^{-1} \bb \\
		= & \frac{\ee_v^\top \LL_{\backslash\mX}^{-2} \ee_v}{\ee_v^\top \LL_{\backslash\mX}^{-1} \ee_v},
	\end{aligned}
\end{equation*}
where the fourth equality and the sixth equality are obtained according to~\eqref{eq:block},
while the fifth equality follows by the cyclicity of trace.
\end{proof}

By~\eqref{eq:block} in the proof of Lemma~\ref{thm:margin}, after adding a single vertex $v$ to $\mX$ forming set $\mX +v$, $\LL_{\backslash(\mX + v)}^{-1}$ can be looked upon as a rank-$1$ update to matrix $\LL^{-1}_{\mX}$:
\begin{equation*}
	\LL_{\backslash(\mX + v)}^{-1} =
	\kh{
		\LL_{\backslash\mX}^{-1} - \frac{\LL_{\backslash\mX}^{-1} \ee_v \ee_v^\top \LL_{\backslash\mX}^{-1}}{\ee_v^\top \LL_{\backslash\mX}^{-1} \ee_v}
		}_{\backslash v},
\end{equation*}
which can be done in time $O(n^2)$ in the case that $\LL^{-1}_{\mX}$ is already computed, rather than directly inverting
a matrix in time $O(n^3)$. This leads to our fast exact algorithm
\textsc{Exact}($\mG$, $\LL$, $k$) to solve Problem~\ref{prob:rdm}, which is outlined in Algorithm~\ref{alg:ExtGreedy}. 
The algorithm first calculates the pseudoinverse  $\LM$ of matrix $\LL$ by using the formula $\LM= \left(\bm{L} - \frac{1}{n}\1\1^{\top}\right)^{-1} + \frac{1}{n}\1\1^{\top}$ obtained in  Proposition~\ref{thm:psuformula} in time $O(n^3)$ and picks a vertex $v_1$ with minimum resistance distance $\Omega(v_1)$, which can be done by using the relation $\Omega(v) = n\LM_{v,v}+\trace{\LM}$ in Proposition~\ref{thm:OmegaiL}. 
Then it works in $k-1$ rounds, each of which includes two main operations. One is to evaluate  $\Delta(\mX,v)$ in $O(n^2)$ time (Line 4), the other is to update $\LL^{-1}_{\backslash\mX}$ in $O(n^2)$ time.
Therefore, the whole running
time of Algorithm~\ref{alg:ExtGreedy} is $O(n^3+kn^2)$, much smaller than $O(kn^4)$.

\begin{algorithm}[tb]
	\caption{\textsc{Exact}$(\mathcal{G},\LL,k)$}\label{alg:ExtGreedy}
	\DontPrintSemicolon
        \Input{
		A weighted digraph $\mG=(\mV,\mE,\WW)$; the Laplacian matrix $\LL$ of $\mG$;
        an integer $1 \leq k \leq |Q|$\\
	}
	\Output{
		$\mX$: A subset of $\mV$ with $\left|\mX\right|=k$
	}
	Compute $\LL^\dag$\;
	Initialize solution $\mX = \setof{v_1}$ where $v_1 = \argmin_{v\in \mV}\nolimits
	n\kh{\LL^\dag}_{v,v} + \trace{\LL^\dag}
	$\;
	\For{$i=2,\ldots,k$}{
        Compute $\Delta(\mX,v)$ for each $v\in\mV\backslash\mX$\;
        Select $v_i$ s.t. $v_i\gets \argmax_{v\in \mV\setminus \mX}\nolimits\Delta(\mX,v)$\;
		Update solution $\mX\gets\mX+v_i$\;
        Update 
		$
			\LL_{\backslash\mX}^{-1} =
			\kh{
				\LL_{\backslash\mX}^{-1} -
				\frac{\LL_{\backslash\mX}^{-1} \ee_{v_i}
					\ee_{v_i}^\top \LL_{\backslash\mX}^{-1}}
				{\ee_{v_i}^\top \LL_{\backslash\mX}^{-1} \ee_{v_i}}}_
			{\backslash v_i}.
		$\;
	}
	\Return{$\mX$}
\end{algorithm}

On the basis of the well-established result in~\cite{NeWoFi78}, Algorithm~\ref{alg:ExtGreedy} provides a $(1 - \frac{k}{k-1}\cdot\frac{1}{e})$-approximation of the optimal solution to Problem~\ref{prob:rdm}.
\begin{theorem}
	The set $\mX$ returned by Algorithm~\ref{alg:ExtGreedy} satisfies
	\begin{equation*}
		\Omega(v^*) -
		\Omega(\mathcal{X})
		\geq 
		\kh{1 - \frac{k}{k-1}\cdot \frac{1}{e}}
		\kh{\Omega(v^*)-
			\Omega(\mathcal{X}^*)
		},
	\end{equation*}
	where $\mX^*$ is the optimal solution to Problem~\ref{prob:rdm} and $v^*$ is the vertex with the minimum resistance distance for a single vertex, i.e.,
\begin{equation*}
\mathcal{X}^* \defeq
	\argmin\limits_{\sizeof{\mathcal{X}}\leq k}
	\Omega(\mathcal{X}),\quad v^* \defeq \argmin\limits_{v\in \mV}
	\Omega(v).
\end{equation*}	
 
\end{theorem}
\begin{proof}
	Let $\mX_i$ be the solution set after exactly $i$ vertices have been selected. By supermodularity, for any $i\geq 1$
	\begin{equation*}
		\Omega(\mX_i) -
		\Omega(\mX_{i+1})
		\geq \frac{1}{k}
		\Omega(\mX_i) -
			\Omega(\mX^*),
	\end{equation*}
	which implies
	\begin{equation*}
		\Omega(\mX_{i+1}) - \Omega(\mX^*) \leq 
		\kh{1 - \frac{1}{k}}
		\kh{ \Omega(\mX_i) -
			 \Omega(\mX^*) }.
	\end{equation*}
	Then, we have
	\begin{equation*}
		 \Omega(\mX) - \Omega(\mX^*) \leq
		 \kh{1 - \frac{1}{k}}^{k-1}
		\kh{ \Omega(\mX_1) - \Omega(\mX^*)} 
		\leq
		 \frac{k}{k-1} \cdot \frac{1}{e}
		\kh{ \Omega(\mX_1) - \Omega(\mX^*)},
	\end{equation*}
	which coupled with
	$\Omega(\mX_1) = \Omega(u^*)$ completes the proof. 
\end{proof}

\subsection{Experiments}

In this subsection, we evaluate the performance of Algorithm~\ref{alg:ExtGreedy} 
 by conducting experiments on three popular model networks and four realistic networks, with the latter taken from KONECT~\cite{Ku13} and SNAP~\cite{LeKr14}. We run our experiments on the largest strongly connected components of these seven networks, related information of which is shown in Table~\ref{table:inf}. All experiments are implemented in Julia, which are run on a Windows desktop with 2.5 GHz Intel i7-11700 CPU and 16G memory, using a single thread.

\begin{table}
\caption{Information of model and real-world networks. For a network with $n$  vertices and $m$ edges, we use $n'$ and $m'$ to denote, respectively, the number of vertices and the number of edges in its largest connected component.}
	\centering
	\setlength{\tabcolsep}{2.3mm}{
	\begin{tabular}{ccccc}
		\toprule
		Network & $n$ & $m$ & $n'$ & $m'$\\
		\midrule
		Watts-Strogatz & $50$ & $500$ & $50$ & $500$\\
		Erd\"os-R\'enyi  & $50$ & $300$ & $50$ & $300$\\
		Scale-Free & $50$ & $282$ & $50$ & $282$\\
		email-Eu-core & $1,005$ & $25,571$ & $803$ & $24,729$\\
		Air traffic control & $1,226$ & $2,615$ &  $792$ & $1,900$\\
		Wiki-Vote & $7,115$ & $103,689$ & $1,300$ & $39,456$ \\
		Advogato & $6,541$ & $51,127$ & $3,140$ & $41,872$\\
		\bottomrule
	\end{tabular}}
 \label{table:inf}
\end{table}

Before presenting our experiment results, we give a brief introduction to the construction of considered model networks and real datasets.  

\textbf{Watts-Strogatz (WS) small-world graph~\cite{WaSt98, Mo03}.} We start with a regular network consisting of $n=50$ vertices connected to their $K=10$ nearest neighbors. The vertices are arranged in a ring with $2K$ edges per vertex. In the rewiring procedure, We select a vertex and the edge that connects this vertex to its nearest neighbor in a counterclockwise sense. With probability $p$ we rewire this edge to a randomly chosen vertex in the network, and with probability $1-p$ we leave it as it is. Self-connections and repeated connections are not allowed. With probability $b$ the edge goes out from the current vertex and with probability $1-b$ it goes into it. We move counterclockwise around the ring and repeat this procedure for each vertex until one lap is completed. Then we repeat the process with the second nearest counterclockwise neighbors of each node, and so on up to the $K$th nearest neighbors. In our experiment, $p$ is set to be $0.5$, and $b$ is set to be $1$.

\textbf{Erd\"os-R\'enyi (ER)  random graph~\cite{ErRe59, Bo01}.} The ER random graph starts with a vertex set $\mV$ with $N$ vertices. Then, for each ordered pair $(i,j)$  satisfying $i,j\in\mV$ and $i\neq j$, we create a directed edge from $i$ to $j$ with probability $p$. In our experiment, $p$ is set to be $0.15$.

\textbf{Scale-Free network (SF)}~\cite{GoKaKi01}: We assign two weights $p_i = i^{-\alpha_{\mathrm{out}}}$ and $q_i = i^{-\alpha_{\mathrm{in}}}$ $(i = 1,\ldots,N)$ to each vertex for outgoing and incoming edges, respectively. Both control parameters $\alpha_{\mathrm{out}}$ and $\alpha_{\mathrm{in}}$ are in the interval $[0, 1)$. Then two different vertices $(i, j)$ are selected with probabilities, $p_i/\sum_k p_k$ and $q_j/\sum_k q_k$, respectively, and an edge from the vertex $i$ to $j$ is created with an arrow pointing to $j$ from $i$. We repeat this process $m$ times. The SF networks generated in this way exhibit the power-law behavior in both outgoing and incoming degree distributions, with their exponents $\gamma_{\mathrm{out}}$ and $\gamma_{\mathrm{in}}$  being $\gamma_{\mathrm{out}} = (1 + \alpha_{\mathrm{out}})/\alpha_{\mathrm{out}}$ and $\gamma_{\mathrm{in}} = (1 + \alpha_{\mathrm{in}})/\alpha_{\mathrm{in}}$, respectively. In our experiment, $\alpha_{\mathrm{in}}$ and $\alpha_{\mathrm{out}}$ are both set to be $0.5$, and $m$ is set to be $300$.

\textbf{email-Eu-core}~\cite{LeKr14, LeKlFa07, YiBeAuLeGl17}: The network was generated using email data from a large European research institution. The vertices in the network represent people within this institution, and there is a directed edge $(u,v)$ from vertex $u$ to vertex $v$ in the network if person $u$ sent person $v$ at least one email. The dataset only includes email communication between members of the institution (the core) and does not include incoming messages from or outgoing messages to external entities.

\textbf{Air traffic control}~\cite{Ku13}: The network was generated from the USA's FAA (Federal Aviation Administration) National Flight Data Center (NFDC), Preferred Routes Database. Vertices in this network represent airports or service centers and directed edges are created from strings of preferred routes recommended by the NFDC.

\textbf{wiki-Vote}~\cite{LeKr14, LeHuDaKl10, LeHuDaKl10/2}: Wikipedia is a free online encyclopedia, created and edited collaboratively by volunteers around the world. A small proportion of Wikipedia contributors are administrators, who have access to additional technical features that aid in maintenance. In order for a user to become an administrator a Request for adminship (RfA) is issued and the Wikipedia community decides who will be granted adminship through a public discussion or a vote. This network contains all the Wikipedia voting data from the inception of Wikipedia till January 2008. Vertices in the network represent Wikipedia users and a directed edge from vertex $i$ to vertex $j$ represents that user $i$ voted on user $j$.

\textbf{Advogato}~\cite{Ku13, MaSaTo09}: This is the trust network of Advogato, an online community platform for free software developers launched in 1999. The vertices in this network represent Advogato users, and the directed edges represent trust relationships. Trust links are referred to as "certifications" on Advogato, and there are three different levels of certifications that correspond to the following three different edge weights: apprentice (0.6), journeyer (0.8), and master (1.0). Users who have not received any trust certifications are referred to as observers. It is possible to trust oneself on Advogato, which leads to the presence of loops in the network.

We now study the accuracy of our algorithm \textsc{Exact} by comparing it with the following baseline strategies for selecting $k$ vertices: \textsc{Optimum}, \textsc{Random}, \textsc{Top-degree} and \textsc{Min-res}. The strategy \textsc{Optimum} selects $k$ vertices with the optimum resistance distance by brute-force search. \textsc{Random} scheme chooses $k$ vertices at random. \textsc{Top-degree} method chooses $k$ vertices with the highest out degrees, while the \textsc{Min-res} approach chooses $k$ vertices with the lowest resistance distance according to Definition~\ref{def:Omegai}.

We first evaluate the effectiveness of algorithm \textsc{Exact} on the three model networks for $k = 1, 2,\ldots, 6$, for which we are able to compute the optimum solutions because of their small sizes. Figure~\ref{fig:syn} presents the resistance distance of vertex sets obtained by different methods, from which we observe that the solutions returned by our algorithm \textsc{Exact} and the optimum solution are almost the same, both of which are better than those returned by the three other baseline schemes. Thus, our algorithm is very effective in practice, the approximation ratio of which is significantly better than the theoretical guarantee. 

\begin{figure}
\centering 
\includegraphics[width=0.48\textwidth]{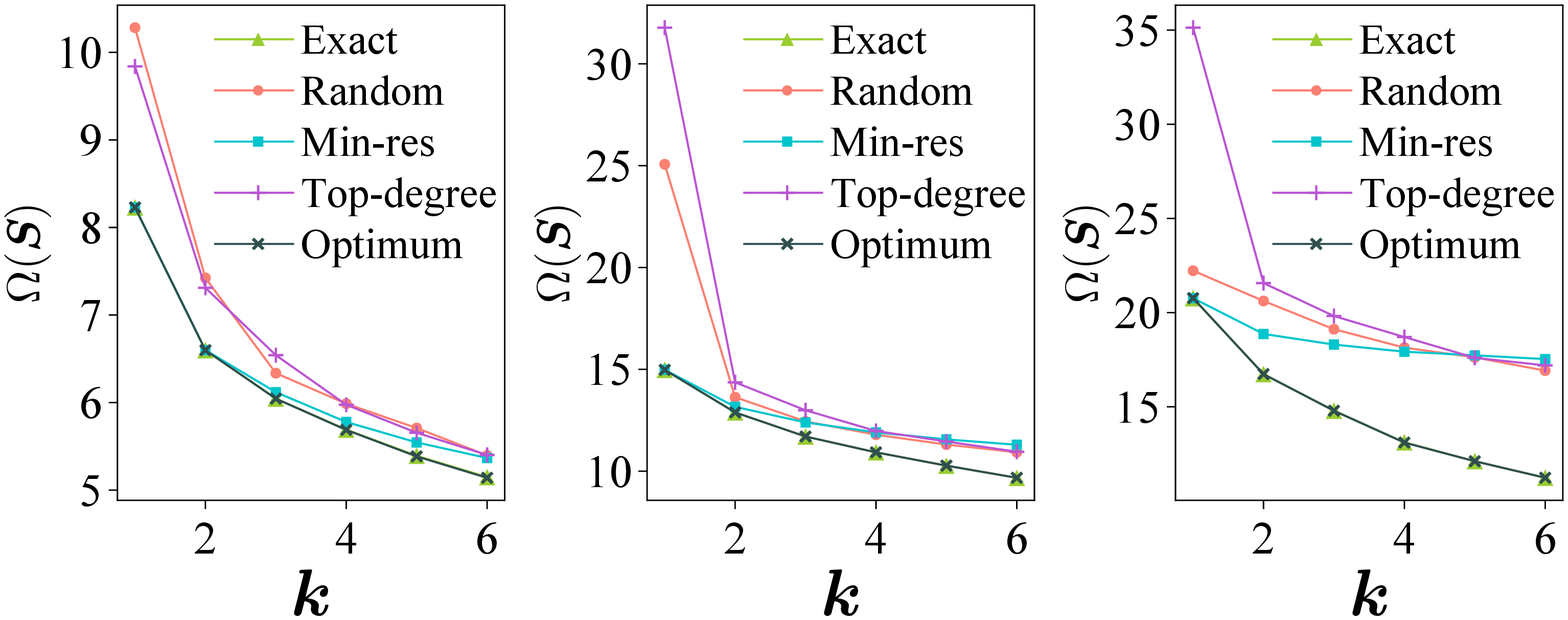} 
\caption{Resistance distance of vertex sets returned by \textsc{Exact}, random and optimum strategies on three models: WS (a), ER (b), and SF (c).} 
\label{fig:syn}
\end{figure}

We then demonstrate the effectiveness of \textsc{Exact} by comparing it with \textsc{Random}, \textsc{Top-degree}, and \textsc{Min-res} on four realistic networks, for which we cannot obtain the optimal solutions. The comparison of
The results for these four strategies are shown in Figure~\ref{fig:real}, which indicates that our greedy algorithm \textsc{Exact} outperforms the three baseline schemes for vertex selection.

\begin{figure}
	\centering 
	\includegraphics[width=0.48\textwidth]{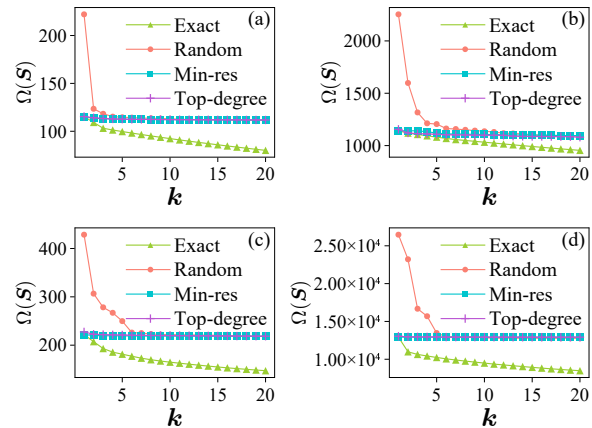} 
	\caption{Resistance distance of vertex sets returned versus the number \(k\) of vertices chosen for the four algorithms on email-Eu-core (a), Air traffic control (b), Wiki-Vote (c), and Advogato (d).} 
	\label{fig:real}
\end{figure}

\section{Conclusion}

We introduced the resistance distances for strongly connected directed graphs based on random walks, which is a natural extension of resistance distances for undirected graphs. We defined the Laplacian matrix $\LL$ for directed graphs, which subsumes the Laplacian matrix of undirected graphs as a special case. We studied the properties of the Laplacian matrix for directed graphs, in terms of whose pseudoinverse we provided expression for the two-node effective resistance, as well as some defined quantities based on effective resistances, such as Kirchhoff index, and multiplicative degree-Kirchhoff index. Moreover, we proved that the two-node resistance distance on directed graphs is a metric. 

In the second part, we defined the resistance distance between a vertex and a vertex group in directed graphs, and expressed this quantity in terms of the elements of the inverse of a submatrix of $\LL$. We further proposed the problem of selecting a set of fixed number of nodes, such that their effective resistance is minimized. Since this combinatorial optimization problem is NP-hard, we presented a greedy algorithm to approximately solve it, which has a proved approximation ratio, since the objective function of the problem is monotone and supermodular. Experiments on model and realistic networks validate the performance of our approximation algorithm. Our work provides useful insight on potential applications of directed graphs in diverse aspects, such as graph clustering, link prediction, and network reliability. 






%
\bibliographystyle{IEEEtran}
\bibliography{directed}

%

\begin{IEEEbiography}{Mingzhe Zhu}
received the B.Sc. degree in the
School of Computer Science, Fudan University,
Shanghai, China, in 2021. He is currently pursuing the Master degree in the School of Computer
Science, Fudan University, Shanghai, China. His
research interests include network science, graph
data mining, and random walks.
\end{IEEEbiography}

\begin{IEEEbiography}{Liwang Zhu}
		received the B.Eng. degree in computer science and technology, Nanjing University of Science and Technology, Nanjing, China, in 2018. He is currently pursuing the Ph.D. degree in the School of Computer Science, Fudan University, Shanghai, China. His research interests include social networks, opinion dynamics, graph data mining and network science.
\end{IEEEbiography}

\begin{IEEEbiography}{Huan Li}
received the B.S. degree and
the M.S. degree in computer science from
Fudan University, Shanghai, China, in 2016
and 2019, respectively. He is currently
pursuing the Ph.D. degree in University
of Pennsylvania. His research interests include graph algorithms, social networks,
and network science.
\end{IEEEbiography}

\begin{IEEEbiography}{Wei Li}
received the B.Eng. degree in automation and
the M.Eng. degree in control science and engineering
from the Harbin Institute of Technology, China, in
2009 and 2011, respectively, and the Ph.D. degree
from the University of Sheffield, U.K., in 2016. After
being a research associate at the University of York,
UK, he is currently an associate professor with the
Academy for Engineering and Technology, Fudan
University. His research interests include robotics
and computational intelligence, and especially selforganized/
swarm systems, and evolutionary machine
learning.
\end{IEEEbiography}

\begin{IEEEbiography}{Zhongzhi Zhang}
	(M'19)	 received the B.Sc. degree in applied mathematics from Anhui University, Hefei, China, in 1997 and the Ph.D. degree in management science and engineering from Dalian University of Technology, Dalian, China, in 2006. \\
	From 2006 to 2008, he was a Post-Doctoral Research Fellow with Fudan University, Shanghai, China, where he is currently a Full Professor with the School of Computer Science. He has published over 160 papers in international journals or conferences. 
 He was selected as one of the most cited Chinese researchers
	(Elsevier) in 2019,  2020, and 2021. His current research interests include network science, graph data mining, social network analysis, computational social science, spectral graph theory, and random walks. \\
	Dr. Zhang was a recipient of the Excellent Doctoral Dissertation Award of Liaoning Province, China, in 2007, the Excellent Post-Doctor Award of Fudan University in 2008, the Shanghai Natural Science Award (third class) in 2013, the CCF Natural Science Award (second class) in 2022, and the Wilkes Award for the best paper published in The Computer Journal in 2019. He is a member of the IEEE.
\end{IEEEbiography}




\end{document}